\documentclass[preprint, numbers]{sigplanconf}
\newif\ifsigalt\sigalttrue

\usepackage{supertech-sig-light}
\usepackage{sb-macros}

\begin{document}

\setlength{\pdfpageheight}{\paperheight}
\setlength{\pdfpagewidth}{\paperwidth}
\setlength{\belowdisplayskip}{0ex} \setlength{\belowdisplayshortskip}{0ex}
\setlength{\abovedisplayskip}{0ex} \setlength{\abovedisplayshortskip}{0ex}
\setlength{\floatsep}{0ex}
\setlength{\intextsep}{0ex}
\setlength{\textfloatsep}{0ex}
\setlength{\abovecaptionskip}{0ex} \setlength{\belowcaptionskip}{1ex}

\conferenceinfo{SPAA'16}{June 27 -- 29, 2016, Asilomar State Beach,California, USA}
\copyrightyear{20yy}
\copyrightdata{978-1-nnnn-nnnn-n/yy/mm} 
\copyrightdoi{nnnnnnn.nnnnnnn}
\publicationrights{licensed}

\title{Extending the Nested Parallel Model to the Nested Dataflow
 Model with Provably Efficient Schedulers}

\authorinfo{David Dinh}
    {Computer Science Division, EECS,
    University of California, Berkeley,
    Berkeley, CA 94720, USA}
    {dinh@cs.berkeley.edu}
\authorinfo{Harsha Vardhan Simhadri}
    {Computer Science Department,
    Lawrence Berkeley National Lab, Berkeley, CA 94720, USA}
    {harshas@lbl.gov}
\authorinfo{Yuan Tang
\thanks{All the coauthors contributed equally to this paper. 
Yuan Tang is the corresponding author. Part of the work was done
when the author was a visiting scientist at MIT CSAIL.}}
    {School of Software, Fudan University, 
    Shanghai Key Lab. of Intelligent Information Processing,
    Shanghai 200433, P. R. China}
    {yuantang@fudan.edu.cn}

\maketitle

\begin{abstract}
The nested parallel (a.k.a. fork-join) model is widely used for writing 
parallel programs. However, the two composition constructs, i.e.
``$\parallel$'' (parallel) and ``$\serial$'' 
(serial), are insufficient in expressing ``partial dependencies'' 
or ``partial parallelism'' in a program.
We propose a new dataflow composition construct ``$\fire$'' to
express partial dependencies in algorithms in a processor- and 
cache-oblivious way, thus extending the Nested Parallel (NP) model
to the \emph{Nested Dataflow} (ND) model.  
We redesign several divide-and-conquer algorithms ranging 
from dense linear algebra to dynamic-programming in the ND model
and prove that they all have optimal span while
retaining optimal cache complexity.
We propose the design of runtime schedulers that map ND programs to
multicore processors with multiple levels of possibly shared caches
(i.e, Parallel Memory Hierarchies~\cite{AlpernCaFe93}) and provide
theoretical guarantees on their ability to preserve locality and load
balance.  For this, we adapt space-bounded (SB) schedulers for the ND
model.  We show that our algorithms have increased
``parallelizability'' in the ND model, and that SB schedulers can use
the extra parallelizability to achieve asymptotically optimal bounds
on cache misses and running time on a greater number of processors
than in the NP model. The running time for the algorithms in this
paper is $O\left(\frac{\sum_{i=0}^{h-1}
  \cc(\Qt;\sigma\cdot M_i)\cdot C_i}{p}\right)$, where $\cc$ is the
cache complexity of task $\Qt$, $C_i$ is the cost of cache miss at
level-$i$ cache which is of size $M_i$, $\sigma\in(0,1)$ is a
constant, and $p$ is the number of processors in an $h$-level cache
hierarchy.
\punt{
For a simpler machine model, by extending the randomized work-stealing
scheduler, we bound the execution time to be $O(T_1/P + T_\infty + 
\log (1/\epsilon))$
with high probability, where $T_1$ and $T_\infty$ are work and
span of the ``enabling tree'' respectively. 
We also prove a cache miss bound of $Q_P = Q_1 + O((P T_\infty + 
N_{\text{sink}}) M/B)$, where $N_{\text{sink}}$ is the number of ``sink'' tasks.
 The bounds of SB scheduler.
} 
\end{abstract}

\vspace{-5pt}

.


\vspace{-5pt}

\secput{intro}{Introduction} A parallel algorithm can be represented
by a directed acyclic graph (DAG) that contains only \defn{data
  dependencies}, but not the control dependencies induced by the
programming model. We call this the \defn{algorithm DAG}.  In an
algorithm DAG, each vertex represents a piece of computation without
any parallel constructs and each directed edge represents a data
dependency from its source to the sink vertex. For example,
\figref{lcs-algo-dag} is the algorithm DAG of the dynamic programming
algorithm for the Longest Common Subsequence (LCS) problem.  This DAG
is a 2D array of vertices labeled $X(i,j)$, where the values with
coordinates $i = 0$ or $j = 0$ are given.  For all $i, j > 0$, vertex
$X(i,j)$ depends on vertices $X(i-1,j-1), X(i,j-1)$ and $X(i-1,j)$.
In an algorithm DAG, there are two possible relations between any pair
of vertices $x$ and $y$. If there is a path from $x$ to $y$ or from
$y$ to $x$, one of them must be executed before the other, i.e. they
have to be serialized; otherwise, the two vertices can run
concurrently. 
 
It is often tedious to specify the algorithm DAG by listing individual
vertices and edges, and in many cases the DAG is not fully known until
the computation has finished. Therefore, higher level programming
models are used to provide a description of a possibly dynamic DAG.
One such model is the \defn{nested parallel programming model} (also
known as fork-join model) in which DAGs can be constructed through
recursive compositions based on two constructs, ``$\parallel$''
(``parallel'') and ``$\serial$'' (``serial''). In the nested parallel
(NP) model, an algorithm DAG can be specified by a \defn{spawn tree},
which is a recursive composition based on these two constructs.  The
internal nodes of the spawn tree are serial and parallel composition
constructs while the leaves are strands --- segments of serial code
that contain no function calls, returns, or {\it spawn} operations. In
a spawn tree, $a\serial b$ is an infix shorthand for a node with
$\serial$ operator and $a$ and $b$ as left and right children and
indicates that $b$ has a dependence on $a$ and thus cannot start until
$a$ finishes, while $a \parallel b$ indicates that $a$ and $b$ can run
concurrently.

\noindent
\begin{figure*}[th]
\vspace{-1em}
\begin{minipage}[b]{0.19 \linewidth}
\subfloat[Algorithm DAG]{
\label{fig:lcs-algo-dag}
\includegraphics[width=0.96\linewidth]{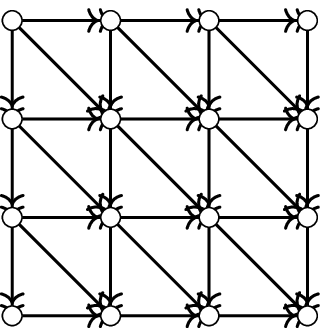}
} 
\end{minipage}
\hfill
\begin{minipage}[b]{0.19 \linewidth}
\subfloat[Divide-and-Conquer.]{ 
\label{fig:lcs-dag-art}
\includegraphics[width=0.96\linewidth]{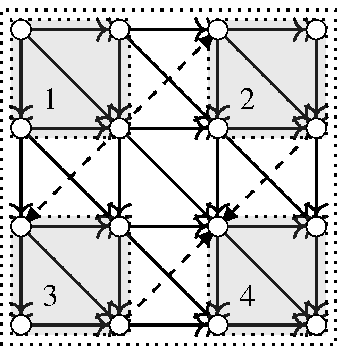}
} 
\end{minipage}
\hfill
\begin{minipage}[b]{0.55 \linewidth}
\subfloat[Spawn Tree]{
\includegraphics[width=0.96\linewidth]{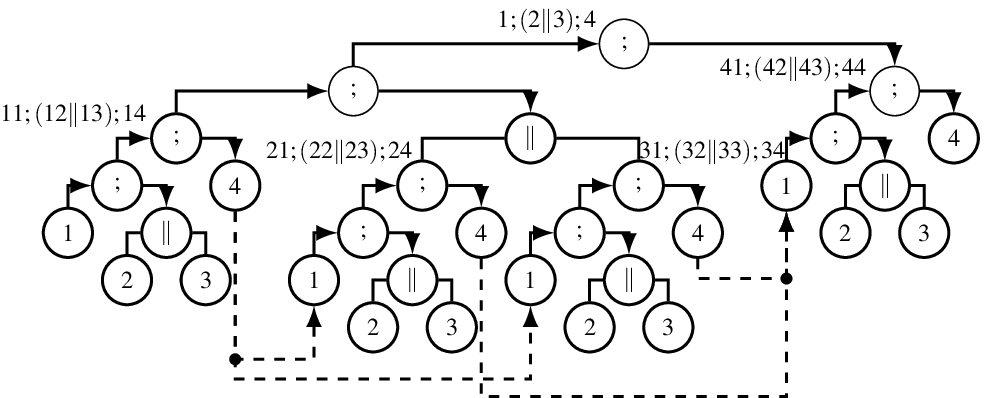}
} 
\end{minipage}
\vspace{0.5ex}
\caption{
Algorithm DAG and the spawn tree of the LCS algorithm in the NP
model. The labels $1,2,3,4$ in the DAG correspond to the four
quadrants in its decomposition. The dashed arrows in the DAG are
artifical dependencies induced when the algorithm DAG is expressed in
the nested parallel model.  The leaves of the spawn tree are smaller
LCS tasks while the internal nodes are composition constructs. The
numerical labels in the spawn tree represent the quadrant of the
dynamic programming table in~\figref{lcs-dag-art} they correspond
to. Solid arrows represent the dataflow indicated by the ``$\serial$''
constrct, and dashed arrows represent artificial dependencies.
\label{fig:lcs-np-tree}
}
\vspace{-3ex}
\end{figure*}

To express the LCS algorithm in the NP model, one might decompose the
2D array of vertices in the algorithm DAG into four smaller blocks,
recursively solve the smaller instances of the LCS algorithms on these
blocks, and compose them to by specifying the dependencies between
them using $\serial$ or $\parallel$ constructs.  \figref{lcs-np-tree}
illustrates the resulting spawn tree up to two levels of
recursion. The NP model demands a serial composition between two
subtrees of the spawn tree even if there is \defn{partial dependency}
between them: that is, a subset of vertices in the DAG corresponding
to one of the subtrees depends on a subset of vertices corresponding
to the other.  As a result, while the spawn tree in the NP programming
model can accurately retain the data dependencies of the algorithm
DAG, it also introduces many \defn{artificial dependencies} that are
not necessary to maintain algorithm correctness.  Artificial
dependencies induced by the NP programming model between subtrees of
the spawn tree in \figref{lcs-np-tree} are shown overlaid by dashed
arrows onto the algorithm DAG in \figref{lcs-dag-art}. Many parallel
algorithms, including dynamic programming algorithms and direct
numerical algorithms, have artificial dependencies when expressed in
the NP programming model than increase the span of the algorithm
(e.g. the span of LCS in the NP model is $O(n\log n)$ as opposed to
$O(n)$ of its algorithm DAG).  \textit{The insufficiency of the nested
  parallel programming model in expressing partial dependencies in a
  spawn tree} is the fundamental reason that causes artificial
dependencies between subtrees of the spawn tree.  This deficiency not
only limits the parallelism of the algorithms exposed to schedulers,
but also makes it difficult to simultaneously optimize for multiple
complexity measures of the program, such as span and cache
complexity~\cite{TangYoKa15}; previous empirical studies on scheduling
NP programs have shown that this deficiency can inhibit effective load
balance~\cite{SimhadriBlFi15}.

\punt{ Therefore, it is natural to ask: \textbf{Is there a programming
    model in which one can write processor- and cache-oblivious
    algorithms without introducing artificial dependencies or
    compromising on locality? Can the programming model be backed by
    provably efficient schedulers?}

\para{Key Observations: } The insufficiency of NP programming model in
expressing ``partial dependency'' in a spawn tree is the fundamental
reason that causes artificial dependency in many divide-and-conquer
algorithms.  The two constructs of the NP model (``$\parallel$'' or
``$\serial$'') are too coarse to precisely express the partial
dependencies that arise in the spawn tree of many algorithms.

The partial dependence patterns that occur in many algorithms
including dynamic programming and direct linear algebra algorithms
have a certain recursive structure that is easy for programmers to
sketch.

To express such patterns, we augment the serial and parallel
constructs with a third $\fire$ (``fire'' ) construct that allows the
specification of finer data dependency patterns between subtrees of a
spawn tree in a recursive manner.  We call this extension the {\bf
  Nested Dataflow} programming model, a natural extension of both the
nested parallel and data parallel programming models.  This paper
describes the semantics of our new $\fire$ composition structure,
demonstrates its usage in the context of some regular algorithms, and
extends known schedulers to the new programming model and presents
bounds on their performance.  }

\para{Our Contributions: }
\begin{itemize}
\item \para{Nested Dataflow model.}  We introduce a new \defn{fire}
  construct, denoted ``$\fire$'', to compose subtrees in a spawn
  tree. This construct, in addition to the $\parallel$ and $\serial$
  constructs, forms the \defn{nested dataflow (ND)} model, an
  extension of the nested parallel programming model. The ``$\fire$''
  construct allows us to precisely specify the partial dependence
  patterns in many algorithms that $\parallel$ and $\serial$
  constructs cannot.  One of the design goals of the ND programming
  model is to allow runtime schedulers to execute inter-processor work
  like a dataflow model, while retaining the locality advantages of
  the nested parallel model by following the depth-first order of
  spawn tree for intra-processor execution.

\item \para{DAG Rewriting System (DRS).} We provide a DAG Rewriting
  System that defines the semantics of the ``$\fire$'' construct by
  specifying the algorithm DAG that is equivalent to a dynamic spawn
  tree in the ND model.
		
\item \para{Re-designed divide-and-conquer algorithms.}  We re-design
  several typical divide-and-conquer algorithms in the ND model
  eliminating artificial dependencies and minimizing span. The set of
  divide-and-conquer algorithms ranges from dense linear algebra to
  dynamic programming, including triangular system solver, Cholesky
  factorization, LU factorization with partial pivoting,
  Floyd-Warshall algorithm for the APSP problem, and dynamic
  programming for LCS.  Our critical insight is that the data
  dependencies in all these algorithm DAGs can be precisely described
  with a small set of recursive partial dependency patterns (which we
  formalize as \textit{sets of fire rules}) that allows us to specify
  them compactly without losing any locality or parallelism. Other
  algorithms such as stencils and fast matrix multiplication can also
  be effectively described in this model.

\begin{figure}[!t]
\includegraphics[width=0.7\columnwidth]{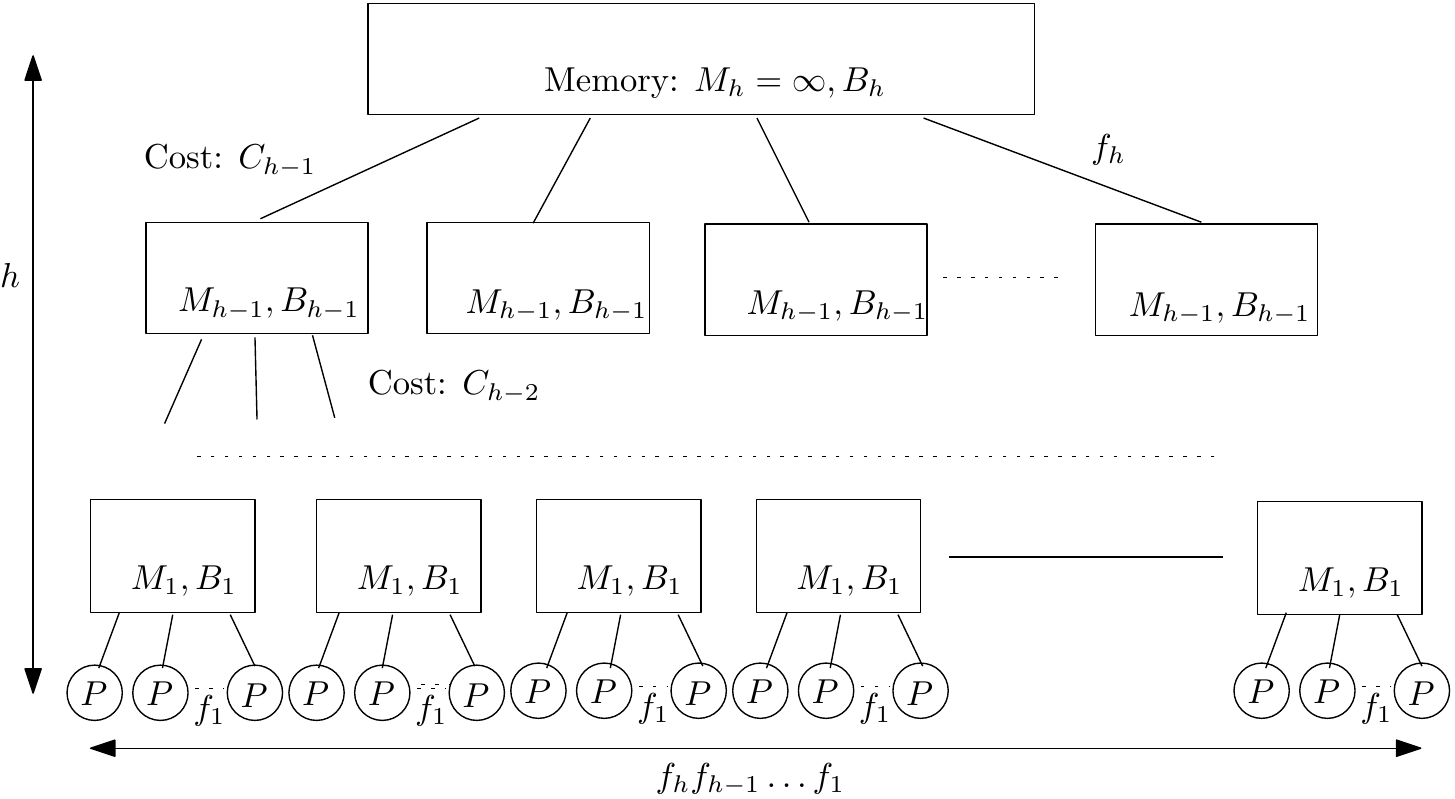}
\caption{An $h$-level parallel memory hierarchy machine model.}
\label{fig:PMH}
\end{figure}

\item \para{Provably Efficient Runtime Schedulers.}  The NP model has
  robust schedulers that map programs to shared-memory multicore
  systems, including those with hierarchical caches
  \cite{BlumofeLe99,ChowdhurySiBl13,BlellochFiGi11}.  These schedulers
  have strong performance bounds for many programs based on complexity
  measures such as \textit{work, span}, and \textit{cache complexity}
  \cite{BlumofeLe93,AroraBlPl98,BlellochGiMa99,Narlikar99b,
    AcarBlBl00,BlellochGi04,BlellochFiGi11,Simhadri13}.  We propose an
  extension of one such class of schedulers called the space-bounded
  schedulers for the ND model and provide provable performance
  guarantees on its performance on the Parallel Memory Hierarchy
  machine model (see~\figref{PMH}).  This machine model accurately
  reflects modern share memory multicore processors in that it has
  multiple levels of possibly shared caches. We show that the
  algorithms in~\secref{num} have greater ``parallelizability'' in the
  ND model than in the NP model, and that the space-bounded schedulers
  can use the extra parallelizability to achieve asymptotically
  optimal bounds on total running time on a greater number of
  processors than in the NP model for ``reasonably regular''
  algorithms. The running time for all the algorithms in this paper is
  asymptotically optimal: $O\left(\frac{\sum_{i=0}^{h-1}
    \cc(\Qt;\sigma\cdot M_i)\cdot C_i}{p}\right)$, where $\cc$ is the
  cache complexity of the algorithm, $C_i$ is the cost of a cache miss
  at level-$i$ cache which is of size $M_i$, $\sigma\in(0,1)$ is a
  constant, and $p$ is the number of processors in an $h$-level cache
  hierarchy.  
  When the input size $N$ is greater than $M_{h-1}$, the size of the
  highest cache level, (below the infinite sized RAM which forms the
  root of the hierarchy), the SB scheduler for the ND model can
  efficiently use all the processors attached to up to
  $N^{1-c}/M_{h-1}$ level-$(h-1)$ caches, where $c$ is an arbitrarily
  small constant.
This compares favorably with the SB scheduler for the NP model
\cite{BlellochFiGi11}, which, for the algorithms in the paper,
requires an input size of at least $M_{h-1}^2$ before it can
asymptotically match the efficiency of the ND version.
  \punt { 
  work-stealing scheduler, we bound the execution time to $O(T_1/P +
  T_\infty + \log (1/\epsilon))$ with high probability, where $T_1$
  and $T_\infty$ are work and span of the ``enabling tree''
  respectively. We also prove a cache miss bound of $Q_P = Q_1 + O((P
  T_\infty + N_{\text{sink}}) M/B)$, where $N_{\text{sink}}$ is the
  number of ``sink'' tasks.  } 
\end{itemize}

\secput{model}{Nested Dataflow Model}
The nested dataflow model extends the NP model by introducing an
additional composition construct, ``$\fire$'', which generalizes the
existing ``$\parallel$'' and ``$\serial$'' constructs. Programs in
both the NP and ND models are expressed as spawn trees, where the
internal nodes are the composition constructs and the leaf nodes are
strands. We refer to subtrees of the spawn tree as \defn{tasks} or
\defn{function calls}. We refer to the subtree rooted at the $i$-th child of an
internal node as its $i$-th \defn{subtask}. In both the models, larger tasks can
be defined by composing smaller tasks with the ``$\serial$'' and ``$\parallel$''
constructs. The ND model allows tasks to be defined as a composition using the
additional \textbf{binary} construct, ``$\fire$'', which enables the specification
of  ``partial dependencies'' between subtasks. This represents an arbitrary 
middle-point between the ``$\serial$'' construct (full dependency) and the
``$\parallel$'' construct (zero dependency).

For example, consider the program represented by the spawn tree
in \figref{code-fire-tree}.  The entire program, $\proc{main}$, is
comprised of two tasks $\proc{F}$ and $\proc{G}$.  Task $\proc{F}$ is
the serial composition of tasks $\proc{a}$ and $\proc{c}$, and task
$\proc{G}$ is the serial composition of $\proc{b}$ and
$\proc{d}$. Task $\proc{C}$ depends on $\proc{A}$, which creates a
partial dependency from $\proc{F}$ to $\proc{G}$. Instead of using a
``$\serial$'' construct, which would block $\proc{d}$ until the
completion of $\proc{F}$ (including both $\proc{a}$ and $\proc{b}$),
we denote the partial dependency with the ``$\fire$'' construct
in \figref{code-fire-tree}.

\noindent
\hspace*{-1cm}
\begin{figure}[!h]
\begin{minipage}[b]{0.27 \linewidth}
\begin{tabular}{l}
$\proc{main}() \{$ \\
$\ \ \ \proc{F}() \tfire{\proc{F}\proc{G}} \proc{G}()$\\
$\}$ \\
\end{tabular}
\end{minipage}%
\begin{minipage}[b]{0.20 \linewidth}
\begin{tabular}{l}
$\proc{F}() \{$ \\
$\ \ \ \proc{a}() \serial \proc{b}()$\\
$\}$\\
\end{tabular}
\end{minipage}%
\begin{minipage}[b]{0.20 \linewidth}
\begin{tabular}{l}
$\proc{G}() \{$ \\
$\ \ \ \proc{c}() \serial \proc{d}()$\\
$\}$\\
\end{tabular}
\end{minipage}%
\begin{minipage}[b]{0.33 \linewidth}
\begin{tabular}{l}
$\oone{+} \tfire{\proc{F}\proc{G}} \oone{-} = \{$ \\
$\ \ \ \otwo{+}{1} \serial \otwo{-}{1}$ \\
$\}$\\
\end{tabular}
\end{minipage}%
\caption{Code for \proc{main}, \proc{F},
\proc{G}, and a fire rule.
\punt{Since there is no recursion in
this example, for any ``$\fire$'' operator at leaf
level, it reduces to a ``$\serial$'' primitive.}
}
\label{fig:code-fire-example}
\end{figure}

\punt{The spawn tree of \figref{code-fire-example} is
shown in \figref{code-fire-tree}. We can see that only
\proc{a} can execute without any dependency. After
\proc{a} is done, both \proc{b} and \proc{c}
are free to execute because they both directly depend on
\proc{a} by ``$\serial$'' primitive and ``$\fire$''
operator respectively. \proc{d} can execute right
after \proc{c}.}

We express this program as code in \figref{code-fire-example}. The partial
dependency from $\proc{F}$ to $\proc{G}$ is specified with the rule
$\tfire{\proc{F}\proc{G}}$. In order to specify that the only dependence 
is from $\proc{a}$, the first subtask of $\proc{F}$, to $\proc{c}$, the
first subtask of $\proc{G}$, we write 
\[
\oone{+}\tfire{\proc{F}\proc{G}}\oone{-} = \{\otwo{+}{1} \serial \otwo{-}{1} \}.
\]
The circled values denote \defn{relative pedigree}, or \defn{pedigree}
in short, which represents the position of a nested function call in a
spawn tree with respect to its ancestor \cite{LeisersonScSu12}. We use
wildcards $\oone{+}$ and $\oone{-}$ to represent the \defn{source}
and \defn{sink} of the partial dependency.  We then specify \defn{a
set of fire rules} to describe the partial dependence pattern of the
``$\fire$'' construct between the source and the sink nodes.  In the
above case, we used $\otwo{+}{1}$ to denote the first subtask of the
source, $\oone{+}$; similarly, $\otwo{-}{1}$ denotes the first subtask
of the sink.  The semicolon indicates a full dependency between
them. In the context of $\proc{main}$ in \figref{code-fire-example},
$\oone{+}$ is bound to $\proc{F}$ and $\oone{-}$ to $\proc{G}$,
implying that there is a full dependency from $\otwo{F}{1}$, which
refers to $\proc{a}$, to $\otwo{G}{1}$, which refers to $\proc{c}$. In
the general case, we allow multiple rewriting rules in the definition
of a fire construct, and ``multilevel'' pedigrees
(e.g. $\othree{+}{2}{1}$ denotes the second subtask of the first
subtask of the source) in each rule.

In the previous example, the dependency from $\proc{a}$ to $\proc{c}$ is a 
full dependency; that is, the entirety of $\proc{a}$ must be completed before 
$\proc{c}$ can start. However, this dependency itself may be artificial.
Therefore, we allow the ``$\fire$'' construct to be recursively defined using 
fire rules that themselves represent partial dependencies.

Consider the following divide-and-conquer algorithm for computing the
matrix product $C += A\times B$, which we denote $\proc{MM}(A,B,C)$.
Let $C_{00}, C_{01}, C_{10}$ and $C_{11}$ denote the top left, bottom left,
top right, and the bottom right quadrants of $C$ respectively.
In the ND model, we can define  $\proc{MM}(A,B,C)$ to be
{\small
\begin{align*}
((\proc{MM}(A_{00},B_{00},C_{00})  &\parallel
\proc{MM}(A_{00},B_{01},C_{01}))
&// \othree{1}{1}{1} || \othree{1}{1}{2}
\\
\parallel ( \proc{MM}(A_{10},B_{00},C_{10})  &\parallel
\proc{MM}(A_{10},B_{01},C_{11})))
&// \othree{1}{2}{1} || \othree{1}{2}{2}
\\
\tfire{MM}
((\proc{MM}(A_{01},B_{10},C_{00})  &\parallel
\proc{MM}(A_{01},B_{11},C_{01})  )
&// \othree{2}{1}{1} || \othree{2}{1}{2}
\\
\parallel ( \proc{MM}(A_{11},B_{10},C_{10})  &\parallel
\proc{MM}(A_{11},B_{11},C_{11}))).
&// \othree{2}{2}{1} || \othree{2}{2}{2}
\end{align*}
}
Each quadrant of $C$ is written to by two of the eight subtasks;
each such pair of subtasks must be serialized to avoid a data race.
For this, we might naively define the fire construct ``$\tfire{MM}$''
between the immediate subtasks of $\proc{MM}(A,B,C)$ with a pair of
fire rules:
\[
\oone{+}\tfire{MM}\oone{-} \quad \{
 \otwo{+}{1} \serial  \otwo{-}{1}, \quad
 \otwo{+}{2} \serial  \otwo{-}{2} \}
\]
However, notice that the dependency between first subtasks (as well as
second subtasks), which is expressed with ``$\serial$'' in the code
above, is in reality a partial dependency. Furthermore, each of these
partial dependencies has the same pattern as ``$\tfire{MM}$''.  Since
this pattern repeats recursively down an arbitrary number of levels,
the ``$\tfire{MM}$'' construct should have been described by the fire
rules:
\begin{equation}
\oone{+}\tfire{MM}\oone{-} \quad \{
 \otwo{+}{1} \tfire{MM} \otwo{-}{1}, \quad
 \otwo{+}{2} \tfire{MM} \otwo{-}{2} \},
\end{equation}
wherein ``$\serial$'' is replaced by ``$\tfire{MM}$''. 

If the recursion terminates at the level indicated
in \figref{MM-tree}, the four instances of ``$\tfire{MM}$'' between
leaves of the spawn tree will be interpreted as four full dependencies
between the corresponding strands.  If the recursion continues, the
fire rules are used to further refine the dependencies.  Whereas this
algorithm has only one set of dependence patterns (fire rules), we
will see algorithms with multiple types of fire rules in the next
section.\\

\punt{Since there can be many ``$\fire$'' operators in an ND program
and ``$\fire$'' operators at different places may want to be rewritten
by different fire rules, ``$\fire$'' operator has to be
overloadable. For the purpose of clarity, we stack a type on top of
``$\fire$'' operator, like ``$\tfire{\proc{F}\proc{G}}$'' in
\figref{code-fire-example}, to denote which fire rule will be
used to rewrite which ``$\fire$'' operator. Since the type may
possibly be inferenced by compiler, an explicit type of fire rule may
not be necessary in a real system. In \figref{code-fire-example}, the
fire rule says that the first child of $\oone{-}$ depends on the first
child of $\oone{+}$. In this case, wildcard $\oone{+}$ matches with
function $\proc{F}$ and $\oone{-}$ matches with function
$\proc{G}$. The ``$\fire$'' construct specifies an recursive data
dependency pattern between its left and right subtasks.  ``$\fire$''
operator can be overloaded for different pairs of source and sink in
the same program and re-defined from program to program, while both
``$\parallel$'' and ``$\serial$'' \defn{primitives} have fixed
semantics and can not be changed.} 

\begin{figure}[t]
\vspace{-2ex}
\begin{minipage}{0.36\columnwidth}
\includegraphics[width=0.95\linewidth]{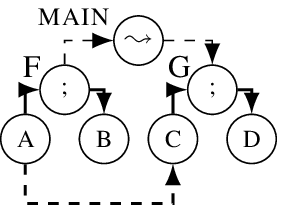}
\vspace{1ex}
\caption{Spawn tree corresponding to the code in \figref{code-fire-example}.}
\label{fig:code-fire-tree}
\end{minipage}
\hfil
\begin{minipage}{0.55\columnwidth}
\includegraphics[width=0.95\linewidth]{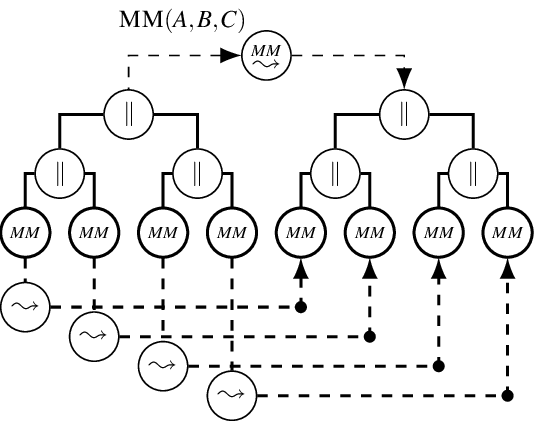}
\caption{Partial dependencies in the recursive matrix multiply algorithm.}
 \label{fig:MM-tree}
\end{minipage}
\end{figure}

\para{DAG Rewriting System (DRS).}
We specify the semantics of the ``$\fire$'' construct with a DRS that
defines the algorithm DAG corresponding to the spawn tree given at
runtime.
The spawn tree can unfold dynamically at runtime by
incrementally \defn{spawning} new tasks -- a spawn operation rewrites
a leaf of the spawn tree into an internal node by adding two new
leaves below. The composition construct in the internal nodes of the
spawn tree imply dependencies between its subtrees.  We represent
these dependencies as directed \defn{dataflow arrows} in the spawn
tree.  The equivalent algorithm DAG implied by the spawn tree is the
DAG with the leaves of the spawn tree as vertices, and edges
representing dataflow edges implied by both the serial and fire
constructs that are incident to the leaves of the spawn tree. The DAG
also grows with the spawn tree; new vertices are added to the DAG
whenever new tasks are spawned, and the construct used in the spawn
operation defines the edges between these new vertices in the
algorithm DAG.  Note that maintaining a full algorithm DAG at runtime
is not necessary. To save space, one can carefully design the order of
the execution of the spawn tree, and recycle the memory used to
represent parts of the spawn tree that have finished executing as
in \cite{BlumofeLe99, LeeBoHuLe10}. We will leave this for future
work.  Instead we focus here on the algorithm DAG to clarify the
semantics of the fire construct.

\punt{
Although the vertices in the Algorithm DAG correspond to leaves in the spawn
tree, we will not expand the entire spawn tree in order to save space. 
Though all leaves of a spawn tree constitute the Algorithm DAG which will
reveal all parallelism, we don't bother to expand the whole tree at
any time during runtime to save space. The research question then becomes 
if it is possible to construct the Algorithm DAG on a dynamically growing 
and shrinking (when some part is done) spawn tree. If we view the root of 
spawn tree as a one-vertex no-edge DAG, with the growing and shrinking of 
tree at runtime, we want a rewriting system to change the DAG accordingly
and always reveal all dependencies of corresponding Algorithm DAG.
}

The DRS iteratively constructs the dataflow edges, and equivalently
the algorithm DAG, by starting with a single vertex representing the
root of the spawn tree and successively applying \defn{DAG rewriting
rules}. Given a DAG $G$, a rewriting rule replaces a sub-graph that is
isomorphic to $L$ with a copy of sub-graph $R=\langle V,E \rangle$,
resulting in a new DAG $G'$. There are two rewriting rules:

\begin{enumerate}
 \item \defn{Spawn Rule}: A spawn rule corresponds to a spawn
        operation. Any current leaf of the spawn tree corresponds to a
        single-vertex no-edge subgraph $L
        = \langle \{\proc{a}\}, \emptyset \rangle$ of the DAG.  If it
        spawns, we rewrite the leaf as a (sub)tree rooted by either a
        ``$\serial$'', ``$\parallel$'' or ``$\fire$'' in the spawn
        tree.\footnote{ A leaf with a non-constant degree parallel
        construct such as a parallel for loop of tasks must be
        rewritten as an binary tree in our programming model.}  The
        root of the newly spawned (sub)tree inherits all incoming and
        outgoing dataflow arrows of the old leaf.
        For instance, if task \proc{a} spawns \proc{b} and \proc{c} in serial, 
        we rewrite the single-vertex, no-edge DAG $L$ to 
        $R = \langle \{\proc{b}, $``$\serial$''$, \proc{c}\},
        \overrightarrow{\proc{b} \proc{c}} \rangle$, where
        $\overrightarrow{\proc{b} \proc{c}}$ is a \defn{solid dataflow arrow} 
        (directed edge) from \proc{b} to \proc{c}
        ($\overrightarrow{\proc{b} \proc{c}}$ is actually a shorthand
        for all-to-all dataflow arrows from all possible descendants of
        \proc{b} to those of \proc{c}, i.e. $\proc{b} \times \proc{c}$).
        If task \proc{a} calls \proc{b} and \proc{c} in parallel, we rewrite it as
        $R = \langle \{\proc{b}, $``$\parallel$''$, \proc{c}\}, \emptyset \rangle$.
        While the parallel construct introduces no dataflow arrows between
        $\proc{b}$ and $\proc{c}$,  a rewriting rule from its closest ancestor 
        that is a ``$\fire$'' construct can introduce dataflow arrows to these
        two nodes according to the \textit{fire rule}.
        The semantics of non-binary serial and parallel constructs are similar. 
        If task \proc{a} invokes ``$\proc{b} \fire \proc{c}$'',
        we rewrite to $R = \langle \{\proc{b} \allowbreak,
        $``$\fire$''$ \allowbreak, \proc{c}\} \allowbreak,
        E' \allowbreak \subseteq \allowbreak\proc{b}
        \allowbreak \times \allowbreak \proc{c} \allowbreak \rangle$,
        where $E'$ is a \defn{dashed dataflow arrow} and is the subset of all
        possible arrows from descendants of \proc{b} to descendants of \proc{c} 
        to be defined by the fire rule as follows.
        \punt{ \figref{spawn-rule} illustrate all three cases of spawn rule.}

  \item \defn{Fire Rule}: Given a dashed dataflow arrow between arbitrary source
        and sink nodes, including those from the left child of a fire construct
        to its right child, we (recursively) rewrite the arrow using the set of
        fire rules associated with it. These rules specify how the ``$\fire$''
        construct is rewritten to a set of dataflow arrows between the 
        descendants of the source and the sink nodes, and their annotations.
        There are two possible cases for rewriting:
        
        \punt{ 
		\figref{fire-rule} illustrated one case of fire rule.
        Other cases of fire rule are similar.
        } 

	\begin{itemize} 

        \item If both operands \proc{a} and \proc{b} are strands, the
            dataflow arrow between them is rewritten as either
            ``$\proc{a} \serial \proc{b}$'' or, if the fire construct
            has no rewriting rules, ``$\proc{a} \parallel \proc{b}$''.

         \item If the source task \proc{a} of a ``$\fire$''
            construct is rewritten by a spawn rule into a tree
            containing $k$ subtasks, we add dataflow arrows
            $E' \subseteq \{\proc{a}_1, \dots\proc{a}_k\} \times \proc{b}$
            to the resulting DAG, i.e. $R = \langle V, E \cup
            E' \rangle$, where the arrows in $E'$ and their labels is
            determined based on the fire rules of the ``$\fire$''
            construct as follows: for a fire rule of the form
            $\otwo{+}{i}p \tfire{T} \oone{-}q$ (where $p$ and $q$ are
            some pedigrees) from $\proc{a}$ to $\proc{b}$, we add a
            dataflow arrow $\oone{+}p \tfire{T} \oone{-}q$ from
            $\proc{a}_i$ to $\proc{b}$.  An analogous rule applies
            when the sink spawns.
            
       \punt{we have $G' = \langle V, E \cup E' \rangle$, where the subset
        $E'$ is a subset of $ \proc{a} \times \{\proc{b}_0, \proc{b}_1\}$.
            
            \item Since the rewrite of a ``$\fire$'' operator adds a subset
            of all-to-all dashed dataflow arrows to the new DAG, a specification
            of fire rule contains a set of fire statements for every added
            dashed dataflow edges.
            \item For a series of chained ``$\fire$'' operators,
            i.e. ``$\proc{a} \fire \proc{b} \fire \proc{c} \ldots$'',
            we rewrite them nestedly from left to right as
            ``$((\proc{a} \fire \proc{b}) \fire \proc{c} \ldots)$''.} 
		\end{itemize}

    \punt{ 
	\item In summary: For spawn rules, we usually rewrite a leaf to a
        (set of) binary tree(s) with new vertices and (solid or
        dashed) arrows added to the DAG; The root of newly added
        (set of) binary tree(s) should inherit all incoming and
        outgoing edges of original leaf vertex; For fire rules,
		we usually just add new dashed arrows.
    } 
\end{enumerate}

From the DRS, it is evident that the binary ``$\serial$'' and ``$\parallel$''
constructs are special cases of the ``$\fire$'' construct. Four fire rules
that recursively refine between both pairs of subtasks of $\oone{+}$ and 
$\oone{-}$ define the $\serial$ construct, and an empty set of rules defines
``$\parallel$''. It is also straightforward to  replace higher-degree 
``$\serial$'' and ``$\parallel$'' constructs using  ``$\fire$'' if one so chooses.

\para{Work-Span Analysis. }
\emph{Work-Span analysis} is commonly used to analyze the complexity of an 
algorithm DAG. We use $T_1$ to denote a task's \defn{work}, that is, the total number of 
instructions it contains. We use $T_\infty$ to denote its \defn{span}, that is, the length 
of the critical path of its DAG.
The composition rule to calculate work $T_1$ for all three constructs 
of the ND model is always a simple summation. For example, if 
$\id{c}=\id{a}\fire\id{b}$, then $T_{1,\id{c}} = T_{1,\id{a}} + T_{1,\id{b}}$.
In principle, the composition to calculate the span $T_\infty$ for all
three constructs is the maximum length of all possible paths from source to sink,
i.e. the critical path. Since ``$\serial$'' and ``$\parallel$'' primitives 
have fixed semantics in all contexts, the span of tasks constructed with them can be simplified
as follows:
for $\id{c} = \id{a} \serial \id{b}$, 
$T_{\infty, \id{c}} = T_{\infty, \id{a}} + T_{\infty, \id{b}}$;
for $\id{c} = \id{a} \parallel \id{b}$, 
$T_{\infty, \id{c}} = \max\{T_{\infty, \id{a}},
T_{\infty, \id{b}}\}$.
On the other hand, since the semantics of a ``$\fire$'' construct
are parameterized by its set of fire rules, we have to calculate the
depth of the task constructed with it on a case-by-case basis.
For instance, for the code in \figref{code-fire-example}, we have
$T_{\infty, \proc{main}} = T_{\infty, \proc{F} \tfire{\proc{F}\proc{G}} \proc{G}}
= \max\{T_{\infty, \proc{a}} + T_{\infty, \proc{b}},
T_{\infty, \proc{a} \serial \proc{c}}+ T_{\infty, \proc{d}}\}$, where
$T_{\infty, \proc{a} \serial \proc{c}}$ is $T_{\infty,\proc{a}}+T_{\infty,\proc{c}}$.
If the rule ``$\tfire{\proc{F}\proc{G}}$'' were to  place a partial
dependence ``$\tfire{\proc{a}\proc{c}}$'' from $\proc{a}$ to $\proc{c}$,
the depth would have to be calculated by further recursive analysis.
\punt{Due to the restricted expressiveness of nested parallel programming model,
we introduce the \emph{nested dataflow} (ND) model. The ND model aims to maximize the
parallelism of inter-processor computation by dataflow execution
and maximize the cache-locality of each processor by preserving
a depth-first execution order with respect to spawn tree simultaneously.
To maximize the inter-processor parallelism, we want the
inter-processor execution to work in dataflow model, i.e.,
any task can start computation as soon as all its data dependencies are satisfied, in a
processor-oblivious way. By processor-oblivious, we mean that the execution 
pattern preserves regardless of how many processor the system has.
To maximize the locality in a cache- and processor-oblivious way, we keep the spawn
tree for a depth-first intra-processor execution order. By cache-oblivious,
we mean that the locality is achieved without any tuning on cache parameters.
If the semantics of ``$\serial$'' primitive indicates an all-to-all dependency or 
\defn{full dependency} (\defn{zero parallelism}) from all descendants of sink (right) 
on source (left) operand and the semantics of ``$\parallel$'' primitive indicates
\defn{zero dependency} (\defn{full parallelism}) between sink and source operands, 
the semantics of ``$\fire$'' operator indicates a \defn{partial dependency}
(\defn{partial parallelism}) between sink and source operands. A partial dependency
at recursion level-$i$ will be recursively defined as a set of partial dependencies
at recursion level-$(i+1)$ up to leaves.} 

\noindent
\begin{figure*}[!ht]
\vspace{-5ex}
\hspace*{-.4cm}
\begin{minipage}[b]{0.45 \linewidth}
\subfloat[Spawn tree of \proc{TRS} with only ``$\parallel$'' and ``$\serial$'' constructs in NP model]{
\label{fig:TRS-old-tree}
\includegraphics[width = 0.98\linewidth]{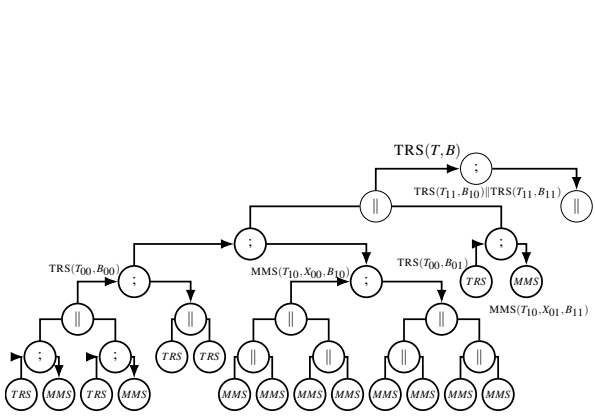}
} 
\end{minipage}
\hfil 
\hspace*{.8cm}
\begin{minipage}[b]{0.45 \linewidth}
\subfloat[Spawn Tree of \proc{TRS} with ``$\fire$'', ``$\parallel$'', and ``$\serial$'' constructs in ND model]{
\label{fig:TRS-new-tree}
\includegraphics[width = \linewidth]{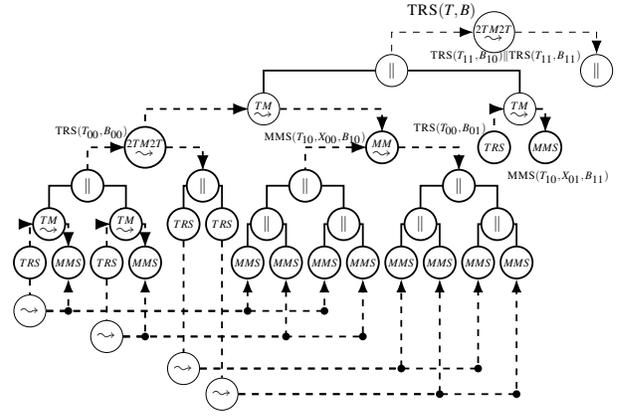}
} 
\end{minipage}
\vspace*{0.5ex}
\caption{Spawn trees of \proc{TRS} in the NP and ND models.  The shape
  of the tree and the leaves are the same between the two models,
  except that some of the $\serial$ constructs in NP model are relaxed
  with $\fire$ constructs and their dataflow arrows in the ND model.
  Dashed arrows corresponding to ``$\fire$'' constructs are recursively
  rewritten until both source and sink subtrees are leaves, where they
  are treated as solid arrows. For simplicity, the figure illustrates 
  only dataflow arrows of type $TM$ between the leaves, and omits 
  dataflow arrows of other types.}
  \vspace{-4ex}
\label{fig:TRS-tree}
\end{figure*}

\secput{num}{Algorithms in the ND Model}
In this section and, we express several typical $2$-way
divide-and-conquer classical linear algebra and dynamic programming
algorithms in the ND model.  These include Triangular System Solver,
Cholesky factorization, LU factorization with partial pivoting,
Floyd-Warshall algorithm for All-Pairs-Shortest-Paths, and LCS. Note
that in going from the NP to the ND model in these algorithms, the
cache complexity of the depth-first traversal will not change as we
leave the divide-and-conquer spawn tree unchanged. At the same time,
we demonstrate that the algorithms have improved parallelism in the ND
model. We do this by proving that their span is smaller than in the NP
model. We will develop more sophisticated metrics to quantify
parallelism in the presence of caches in~\secref{sb-sched}; it turns
out those metrics demonstrate improved parallelism in the ND model as
well.

\para{Triangular System Solver.}
We begin with the Triangular System Solver (\proc{TRS}). $\proc{TRS}(T, B)$
takes as input a lower triangular $n\times n$ matrix $T$ and a square
matrix $B$ and outputs a square matrix $X$ such that $TX = B$. A triangular
system can be recursively decomposed as shown in \eqref{TRS-matrix}.
\begin{figure}[!ht]
\begin{align}
    \begin{bmatrix}
        B_{00} & B_{01} \\
        B_{10} & B_{11}
    \end{bmatrix}
    &=
    \begin{bmatrix}
        T_{00} & 0 \\
        T_{10} & T_{11}
    \end{bmatrix}
    \begin{bmatrix}
        X_{00} & X_{01} \\
        X_{10} & X_{11}
    \end{bmatrix} \nonumber \\
    &=
    \begin{bmatrix}
        T_{00} X_{00} & T_{00} X_{01} \\
        T_{10} X_{00} + T_{11} X_{10} & T_{10} X_{01} + T_{11} X_{11}
    \end{bmatrix}
\label{eq:TRS-matrix}
\end{align}
\end{figure}

\punt{ 
\begin{figure}[!ht]
\begin{tikzpicture}
[inner sep=0mm
, minimum size=4.8mm
, level distance=6mm,
, every node/.style=tree node
, every label/.append style={draw=none, xshift=-1mm, yshift=-2mm}
]
\node[label={above left, xshift=-1mm, yshift=-1mm:$\proc{TRS}(T,B)$}](r){$\serial$}
[sibling distance=32mm]
child{node[]{$\parallel$}
    [sibling distance=40mm]
    child{node[]{$\serial$}
        [sibling distance=37mm]
        child{node[label={above left, xshift=-1mm, yshift=-2mm:\tiny{$\proc{TRS}(T_{00},B_{00})$}}]{$\serial$}
            [sibling distance=18mm]
            child{node[]{$\parallel$}
                [sibling distance=12mm]
                child{node[]{$\serial$}
                    [sibling distance=6mm]
                    child{node[]{\tiny{$TRS$}}
                    edge from parent[ser-in]}
                    child{node[]{\tiny{$MMS$}}
                    edge from parent[ser-out]}
                edge from parent[par-in]}
                child{node[]{$\serial$}
                    [sibling distance=6mm]
                    child{node[]{\tiny{$TRS$}}
                    edge from parent[ser-in]}
                    child{node[]{\tiny{$MMS$}}
                    edge from parent[ser-out]}
                edge from parent[par-out]}
            edge from parent[ser-in]}
            child{node[]{$\parallel$}
                [sibling distance=6mm]
                child{node[]{\tiny{$TRS$}}
                edge from parent[par-in]}
                child{node[]{\tiny{$TRS$}}
                edge from parent[par-out]}
            edge from parent[ser-out]}
        edge from parent[ser-in]}
        child{node[label={above left, xshift=-2mm, yshift=-4mm:\tiny{$\proc{MMS}(T_{10},X_{00},B_{10})$}}]{$\serial$}
            [sibling distance=24mm]
            child{node[tree node]{$\parallel$}
                [sibling distance=12mm]
                child{node[tree node]{$\parallel$}
                    [sibling distance=6mm]
                    child{node[tree node]{\tiny{$MMS$}}
                    edge from parent[par-in]}
                    child{node[tree node]{\tiny{$MMS$}}
                    edge from parent[par-out]}
                edge from parent[par-in]}
                child{node[tree node]{$\parallel$}
                    [sibling distance=6mm]
                    child{node[tree node]{\tiny{$MMS$}}
                    edge from parent[par-in]}
                    child{node[tree node]{\tiny{$MMS$}}
                    edge from parent[par-out]}
                edge from parent[par-out]}
            edge from parent[ser-in]}
            child{node[tree node]{$\parallel$}
                [sibling distance=12mm]
                child{node[tree node]{$\parallel$}
                    [sibling distance=6mm]
                    child{node[tree node]{\tiny{$MMS$}}
                    edge from parent[par-in]}
                    child{node[tree node]{\tiny{$MMS$}}
                    edge from parent[par-out]}
                edge from parent[par-in]}
                child{node[tree node]{$\parallel$}
                    [sibling distance=6mm]
                    child{node[tree node]{\tiny{$MMS$}}
                    edge from parent[par-in]}
                    child{node[tree node]{\tiny{$MMS$}}
                    edge from parent[par-out]}
                edge from parent[par-out]}
            edge from parent[ser-out]}
        edge from parent[ser-out]}
    edge from parent[par-in]}
    child{node[]{$\serial$}
        [sibling distance=8mm]
        child{node[label={above left, yshift=-1mm:\tiny{$\proc{TRS}(T_{00},B_{01})$}}]{\tiny{$TRS$}}
        edge from parent[ser-in]}
        child{node[label={below, xshift=3mm, yshift=8mm:\tiny{$\proc{MMS}(T_{10},X_{01},B_{11})$}}]{\tiny{$MMS$}}
        edge from parent[ser-out]}
    edge from parent[par-out]}
edge from parent[ser-in]}
child{node[label={above left, xshift=-3mm, yshift=-6mm:\tiny{$\proc{TRS}(T_{11},B_{10}) \| \proc{TRS}(T_{11},B_{11})$}}]{$\parallel$}
edge from parent[ser-out]};
\end{tikzpicture}
\caption{$2$-way divide-and-conquer \proc{TRS} algorithm}
\label{fig:TRS-algo}
\end{figure}
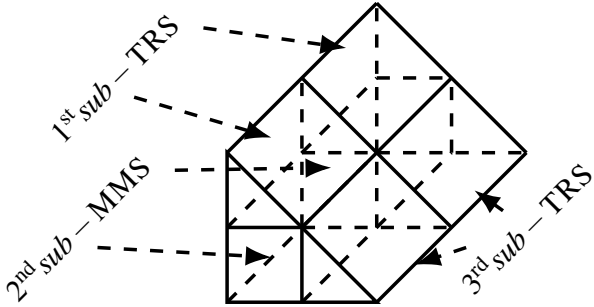
} 

\eqref{TRS-matrix} recursively solves \proc{TRS} on four equally sized sub-quadrants
$X_{00}$, $X_{01}$, $X_{10}$, and $X_{11}$, as graphically depicted in 
\figref{2-way-TRS-algo}. It can be expressed in the NP model as shown in
\eqref{TRS-NP}, where $\proc{MMS}(A, B, C)$ represents a cache-oblivious matrix 
multiplication and subtraction (identical to the one presented in~\secref{model}, except 
instead of computing $C += AB$ it computes $C -= AB$) with span $O(n)$ and using 
$O(n^2)$ space.\footnote{There is also an $8$-way divide-and-conquer cache-oblivious 
parallel algorithm of \proc{MMS} that has an optimal span of $O(\log^2 n)$ but uses 
$O(n^3)$  space which can be used to trade off span for space complexity.}

\begin{figure}[!ht]
\begin{align}
    X &\assign \proc{TRS}(T, B) = \nonumber \\
      (( & X_{00} \assign \proc{TRS}(T_{00}, B_{00}) \serial \proc{MMS}(T_{10}, X_{00}, B_{10})) \nonumber \\
      & \parallel (X_{01} \assign \proc{TRS}(T_{00}, B_{01}) \serial \proc{MMS}(T_{10}, X_{01}, B_{11}) )) \nonumber \\
     \serial& (X_{10} \assign \proc{TRS}(T_{11}, B_{10}) \parallel X_{11} \assign \proc{TRS}(T_{11}, B_{11}))
\label{eq:TRS-NP}
\end{align}
\end{figure}

The span of the TRS algorithm, expressed in the NP model is given by
the recurrence $T_{\infty, \proc{TRS}}(n) = 2 T_{\infty, \proc{TRS}}(n/2) 
+ T_{\infty, \proc{SMM}}(n/2)$, which evaluates to $O(n \log n)$, and is 
not optimal; a straightforward right-looking algorithm has a span of $O(n)$.

\begin{figure}[!h]
\includegraphics[width = \linewidth]{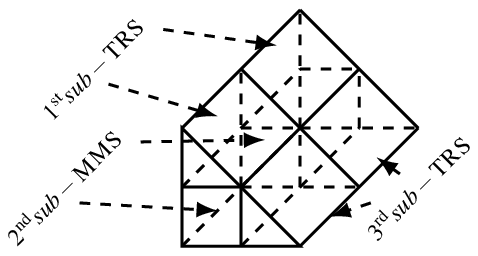}
 \caption{$2$-way divide-and-conquer \proc{TRS} algorithm}
\label{fig:2-way-TRS-algo}
\end{figure}


In \eqref{TRS-new-sched}, we replace the ``$\serial$'' constructs
from the original schedule with ``$\fire$'' construct, in order to
remove artificial dependencies. Because the two ``$\fire$'' constructs
join different types of tasks, they have distinct types, which we 
denote ``$\tfire{TM}$'' and ``$\tfire{2TM2T}$''. Note that
there are algorithms, e.g. Cholesky factorization, where two types
of subtasks, say $\proc{trs}$ and $\proc{mms}$, have more than one kind
of partial dependency pattern between them based on where they occur.
Each type of fire construct has different set of fire rules;
in order to determine what these rules are, we expand an additional
level of recursion to examine finer-grained data dependencies.
\begin{align}
    X &\assign \proc{TRS}(T, B) = \nonumber \\
      (( & X_{00} \assign \proc{TRS}(T_{00}, B_{00}) \tfire{TM} \proc{MMS}(T_{10}, X_{00}, B_{10})) \nonumber \\
      & \parallel (X_{01} \assign \proc{TRS}(T_{00}, B_{01}) \tfire{TM} \proc{MMS}(T_{10}, X_{01}, B_{11}) )) \nonumber \\
     \tfire{2TM2T}& (X_{10} \assign \proc{TRS}(T_{11}, B_{10}) \parallel X_{11} \assign \proc{TRS}(T_{11}, B_{11}))
\label{eq:TRS-new-sched}
\end{align}

Notice that the source task of $\tfire{2TM2T}$ is $((
X_{00} \assign \proc{TRS}(T_{00},
B_{00}) \tfire{TM} \proc{MMS}(T_{10}, X_{00}, B_{10})) \parallel
(X_{01} \assign \proc{TRS}(T_{00},
B_{01}) \tfire{TM} \proc{MMS}(T_{10}, X_{01}, B_{11}) ))$, and the its
sink is $(X_{10} \assign \proc{TRS}(T_{11}, B_{10}) \parallel
X_{11} \assign \proc{TRS}(T_{11}, B_{11}))$.  Since the left subtask
of the sink can start as soon as the matrix multiply updating $B_{10}$
(which is the right subtask of the left subtask of the sink) is
completed, and the right subtask of the sink analogously depends on
the matrix multiply updating $B_{11}$, the fire rule for
$\tfire{2T2M}$ is simply:
\begin{equation}
\oone{+} \tfire{2TM2T} \oone{-} =
 \{ \othree{+}{1}{2} \tfire{MT} \otwo{-}{1}, \othree{+}{2}{2} \tfire{MT} \otwo{-}{2} \}.
\label{eq:def2T2M}
\end{equation}
Both the fire constructs in the fire rules are of type
``$\tfire{MT}$'' since the dependency structure is identical: the
matrix updated in the source \proc{MMS} operation is used as a
dependency in the second argument of the \proc{TRS} operation.

In order to determine the set of fire rules for ``$\tfire{MT}$'', we
expand a pair of subtasks connected by the ``$\tfire{MT}$'' construct
to an additional level of recursion.  For instance, we will expand the
task $\proc{MMS}(T_{10}, X_{00}, B_{10})$ in
equation \eqref{TRS-new-sub-sched}, which (as the source) binds to
$\oone{+}$ in $\tfire{TM}$, and $X_{00} \assign \proc{TRS}(T_{11},
B_{10})$ in \eqref{MM-sub-sched}, which binds to $\oone{-}$. In the
following program, we use $A_{00,11}$ to denote the bottom right
quadrant of the top left quadrant of $A$.

\begin{align}
&\proc{MMS}(T_{10}, X_{00}, B_{10}) =
\quad \quad \quad \quad \quad \quad \quad \quad \quad \quad \quad \quad \quad \quad \quad \  // \oone{+} \nonumber \\
&((\proc{MMS}(T_{10,00},X_{00,00},B_{10,00})  \parallel
\proc{MMS}(T_{10,00},X_{00,01},B_{10,01}))
\nonumber \\
&\parallel ( \proc{MMS}(T_{10,10},X_{00,00},B_{10,10}) \parallel 
\proc{MMS}(T_{10,10},X_{00,01},B_{10,11})))
\nonumber \\
&\tfire{MM}
((\proc{MMS}(T_{10,01},X_{00,10},B_{10,00}) \parallel 
\proc{MMS}(T_{10,01},X_{00,11},B_{10,01}))
\nonumber \\
&\parallel ( \proc{MMS}(T_{10,11},X_{00,10},B_{10,10}) \parallel
\proc{MMS}(T_{10,11},X_{00,11},B_{10,11}))).
\label{eq:MM-sub-sched}
\end{align}

\noindent
\begin{flalign}
 X_{10} & \assign \proc{TRS}(T_{11}, B_{10}) =
\quad \quad \quad \quad \quad \quad \quad \quad \quad \quad \quad \quad // \oone{-} \nonumber \\
      (( & X_{00,00} \assign \proc{TRS}(T_{11,00}, B_{10,00}) \tfire{TM}
      \proc{MMS}(T_{11,10}, X_{00,00}, B_{10,10})) \nonumber \\
      & \parallel (X_{00,01} \assign \proc{TRS}(T_{11,00}, B_{10,01})\tfire{TM}
      \proc{MMS}(T_{11,10}, X_{00,01}, B_{10,11}) )) \nonumber \\
     \tfire{2TM2T}& (X_{00,10} \assign \proc{TRS}(T_{11,11}, B_{10,10}) \parallel
     X_{00,11} \assign \proc{TRS}(T_{11,11}, B_{10,11}))
\label{eq:TRS-new-sub-sched}
\end{flalign}

The dependence of the sink task, $\oone{-}$, on the source task, $\oone{+}$, 
in ``$\tfire{MT}$'' is a result of requiring the value of matrix $B_{10}$ to be
updated by $\oone{+}$ before $\oone{-}$ can use it in a computation. 
At a more fine-grained level, we can examine which quadrant of $B_{10}$ 
each subtask of $\oone{-}$ requires (and which subtask of $\oone{+}$ computes
that quadrant) in order to calculate the fine-grained dependencies. 
For instance, consider the subtask $X_{00,00} \assign \proc{TRS}(T_{11,00}, 
B_{10,00})$, whose pedigree is $\ofour{-}{1}{1}{1}$, which requires $B_{10,00}$.
This quadrant of $B_{10}$ is updated in $\proc{MMS}(T_{10,01},X_{00,10},B_{10,00})$ 
of the source, whose pedigree is $\ofour{+}{2}{1}{1}$. Furthermore, notice that
the dependency from $\ofour{+}{2}{1}{1}$ to $\ofour{-}{1}{1}{1}$ takes the
same form as the dependency from $\oone{+}$ to $\oone{-}$ itself: the matrix 
updated by the source is the second argument in the sink. Therefore, the fire 
rule for this particular dependency is 
$\ofour{+}{2}{1}{1} \tfire{TM} \ofour{-}{1}{1}{1}$.
Similar reasoning gives the remaining fire rules:
\begin{align}
    \oone{+} \tfire{TM} \oone{-} =& \{ 
    	\ofour{+}{1}{1}{1} \tfire{TM} \ofour{-}{1}{1}{1},
    	\ofour{+}{1}{1}{1} \tfire{TM} \ofour{-}{1}{2}{1}, \nonumber \\
     &\ \ofour{+}{1}{2}{1} \tfire{TM} \ofour{-}{1}{1}{2},
    	\ofour{+}{1}{2}{1} \tfire{TM} \ofour{-}{1}{2}{2}, \nonumber \\
     &\ \othree{+}{2}{1} \tfire{TM} \ofour{-}{2}{1}{1},
    	\othree{+}{2}{1} \tfire{TM} \ofour{-}{2}{2}{1}, \nonumber \\
     &\ \othree{+}{2}{2} \tfire{TM} \ofour{-}{2}{1}{2},
    	\othree{+}{2}{2} \tfire{TM} \ofour{-}{2}{2}{2} \}    
\label{eq:TM}
\end{align}
 
\begin{align*} 
	\oone{+} \tfire{2TM2T} \oone{-} =& \{
		\othree{+}{1}{2} \tfire{MT} \otwo{-}{1}, \othree{+}{2}{2} \tfire{MT} \otwo{-}{2} \}
    \\
    \oone{+} \tfire{MT} \oone{-} =& \{
       \ofour{+}{2}{1}{1} \tfire{MM} \ofour{-}{1}{1}{2}, 
       \ofour{+}{2}{1}{2} \tfire{MM} \ofour{-}{1}{2}{2}, \\
     & \ \ofour{+}{2}{2}{1} \tfire{MT} \ofour{-}{1}{1}{1}, 
       \ofour{+}{2}{2}{2} \tfire{MT} \ofour{-}{1}{2}{1} \}
\end{align*} 

\begin{wrapfigure}{R}{0.28\columnwidth}
\includegraphics[width = 0.28\columnwidth]{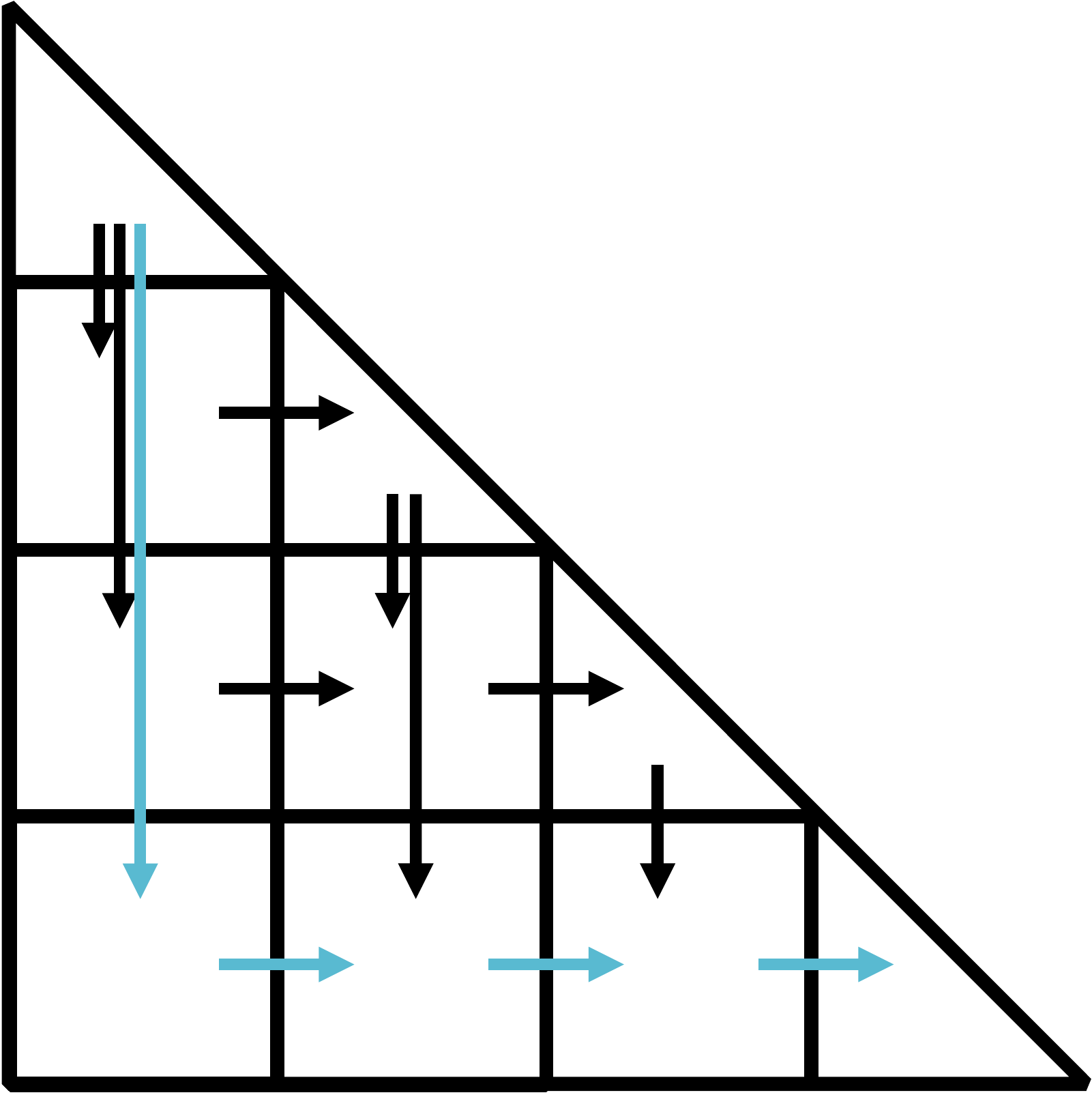}
  \caption{TRS DAG cross-section}
  \label{fig:trs-lattice-dag}
\end{wrapfigure}

We now argue that the span of TRS in the ND model is $O(n)$. The span
of an algorithm is the length of the longest path in its DAG. The
algorithm DAG defined by TRS expressed in the ND model forms a
periodic structure, a cross section of which can be seen in
\figref{trs-lattice-dag}, where squares represent matrix
 multiplications and triangles represent smaller TRS tasks (there are
no edges between separate cross sections).  The length of the longest
path in the DAG, shown in blue, is $O(n)$.

We now formally prove that the algorithm we constructed in the ND
model achieves this span. Let $T_{\infty,\proc{TRS}}(n)$ denote the
span of $\proc{TRS}$ on a matrix with input size $n\times n$, and let
$T_{\infty, \proc{TRS}\ p}(n)$ denote the span of the subtask with
pedigree $p$ descended from $\proc{TRS}$ with input size $n\times n$.
Furthermore, let $T_{\infty,\tfire{TM}}(n)$ denote the critical path
length of a TRS composed with a matrix multiply by a ``$\tfire{TM}$''
construct, where both tasks are directly descended from a TRS of size
$n\times n$. Note that it involves two tasks, a $\proc{TRS}$ and
$\proc{MMS}$, of size $n/2 \times n/2$ each.

Since replacing a fire construct with a serial construct can only
increase the span, it suffices to show that a version of the problem
with some fire constructs replaced with serial constructs has optimal
span. In this analysis, we will replace the ``$\tfire{2TM2T}$''
construct with a serial composition, giving the following upper bound
on the span of TRS in the ND model:
\begin{equation}
\label{eq:trs-depth}
T_{\infty, \proc{TRS}}(n) \le
T_{\infty, \proc{TRS}\oone{1}}(n) 
+ T_{\infty, \proc{TRS}\oone{2}}(n).
\end{equation}

Since the right subtask of \proc{TRS} is merely the parallel
composition of two \proc{TRS} operations, each on a matrix of size
$n/2\times n/2$, the second term on the right reduces to the max of
their (identical) spans, which is
\[
T_{\infty, \proc{TRS}\oone{2}}(n) = T_{\infty, \proc{TRS}}(n/2).
\]

The left subtask consists of two pairs (connected by a parallel
composition), each consisting of a \proc{TRS} task and a \proc{MMS}
task, connected by ``$\tfire{TM}$'' construct and done in parallel,
and their spans are identical.  Therefore, the first term on the right
hand side of inequality~\ref{eq:trs-depth} reduces to
\[
T_{\infty, \proc{TRS}\oone{1}}(n) =
T_{\infty, \proc{TRS}\othree{1}{1}{1} \tfire{TM} \proc{TRS}\othree{1}{1}{2}}(n) =
T_{\infty, \tfire{TM}}(n).
\]

The term on the right is the maximum length among all possible paths
rewritten from
$\othree{1}{1}{1} \tfire{\proc{T}\proc{M}} \othree{1}{1}{2}$.  There
are two types of paths are could potentially be the longest. An
instance of the first type is the ``$\tfire{TM}$'' composition of
tasks TRS\osix{1}{1}{1}{1}{1}{1}, a TRS of size $n/4$, with
TRS\osix{1}{1}{2}{1}{1}{1}, a MMS of size $n/4$, followed by a MMS of
size $n/4$. This gives the first expression in the max term in the
equation below. An instance of the second type is the $\tfire{TM}$
composition of the task $\proc{trs}\osix{1}{1}{1}{1}{1}{1}$, a TRS of
size $n/4$, with $\proc{trs}\osix{1}{1}{1}{1}{1}{2}$, a MMS of size
$n/4$, followed by the ``$\tfire{TM}$'' composition of
$\proc{trs}\ofive{1}{1}{1}{2}{1}$, a TRS of size $n/4$, with
$\proc{trs}\osix{1}{1}{2}{2}{1}{1}$, a MMS of size $n/4$. This results
in the second expression in the max term below.
\[
T_{\infty, \tfire{\proc{T}\proc{M}}}(n) \le
\max\{T_{\infty, \tfire{\proc{T}\proc{M}}}(n/2) + T_{\infty, \proc{MMS}}(n/4), \hspace{2pt} 
2 T_{\infty, \tfire{\proc{T}\proc{M}}}(n/2)\}\nonumber
\]

For the base case of the recurrence, we simply run \proc{TRS} and
\proc{MM} sequentially at the base case size. Therefore, we have
\[T_{\infty, \tfire{\proc{T}\proc{M}}}(1) =
T_{\infty, \proc{TRS}}(1) + T_{\infty, \proc{MMS}}(1) = O(1).\]

Noting that $T_{\infty,\proc{MMS}}(n) = O(n)$, the recurrences can be
solved to show that $T_{\infty, \proc{TRS}}(n) = O(n)$, which is
asymptotically optimal.

\punt{
---------------------------
Since the span of a ``$\fire$'' operator is calculated as a
``$\max$'' operation over all possible paths from source to
sink, \eqref{TRS-fire} specifies how a path of
``$\tfire{\proc{T}\proc{M}}$'' at
recursion level-$i$ can be rewritten to a set of
``$\tfire{\proc{T}\proc{M}}$'' paths at recursion
level-$(i+1)$ with respect to ``$\parallel$'' and ``$\serial$''
primitives in \eqreftwo{TRS-new-sub-sched}{MM-sub-sched},
respectively.
we have a new span recurrence of ND \proc{TRS} algorithm
shown in \eqref{TRS-new-span-rec},
following the paths illustrated in \eqref{TRS-fire-paths}.
The paths in \eqref{TRS-fire-paths} follows either
``$\fire$'' operators or ``$\serial$'' primitives in equation
or explicit arrows rewritten by fire rules in \eqref{TRS-fire}.
\begin{figure}[!ht]
\begin{align}
    \Comment \quad T_{\infty, \proc{TRS}} &= T_{\infty, \oone{1} \tfire{\proc{T}\proc{M}} \oone{2}} + T_{\infty, \oone{5}} \nonumber \\
    T_{\infty, \proc{TRS}}(n) &= T_{\infty, \tfire{\proc{T}\proc{M}}}(n/2) + T_{\infty, \proc{TRS}}(n/2) \nonumber \\
    \Comment T_{\infty, \oone{1} \tfire{\proc{T}\proc{M}} \oone{2}} &= \max\{T_{\infty, \otwo{1}{1} \tfire{\proc{T}\proc{M}} \otwo{2}{1}} + T_{\infty, \otwo{2}{5}} \nonumber\\
    \Comment \quad \quad \quad \quad & \quad \quad \quad, T_{\infty, \otwo{1}{1} \tfire{\proc{T}\proc{M}} \otwo{1}{2}} + T_{\infty, \otwo{1}{5} \tfire{\proc{T}\proc{M}} \otwo{2}{5}}\} \nonumber \\
    T_{\infty, \tfire{\proc{T}\proc{M}}}(n) &= \max\{T_{\infty, \tfire{\proc{T}\proc{M}}}(n/2) + T_{\infty, \proc{MMS}}(n/2), \hspace{2pt} 2 T_{\infty, \tfire{\proc{T}\proc{M}}}(n/2)\}\nonumber \\
    T_{\infty, \tfire{\proc{T}\proc{M}}}(1) &= T_{\infty, \proc{TRS}}(1) + T_{\infty, \proc{MMS}}(1) \label{eq:TRS-new-span-rec}
\end{align}
\end{figure}
The first recurrence of \eqref{TRS-new-span-rec} comes directly
from the top level of recursion in \eqref{TRS-new-sched}. It
should be rigorously calculated as
$\max\{\oone{1} \tfire{\proc{T}\proc{M}} \oone{2},
\oone{3} \tfire{\proc{T}\proc{M}} \oone{4}\} +
\max\{\oone{5}, \oone{6}\}$. Since we only parallelize the
same function computed on the same sizes of inputs, we
omit the ``$\max$''
operation for ``$\parallel$'' primitive here to save space.
The second recurrence takes a ``$\max$'' operation over all
possible paths rewritten by \eqref{TRS-fire} from
$\oone{1} \tfire{\proc{T}\proc{M}} \oone{2}$ with respect to
``$\parallel$'' and ``$\serial$'' primitives.
(refer to \eqref{TRS-fire-paths} for all possible paths from
$\oone{1}$ to $\oone{2}$)
The last equation in \eqref{TRS-new-span-rec} says that when both
operand subtrees are leaves, the ``$\tfire{\proc{T}\proc{M}}$''
operator reduces to ``$\serial$'' primitive. This span recurrence
solves to $O(n)$, recalling that
$T_{\infty, \proc{MMS}}(n) = O(n)$, which is asymptotically
optimal.


\figreftwo{TRS-old-tree}{TRS-new-tree} are the spawn trees of
\proc{TRS} in nested parallel model and ND model, respectively.
Compared \figref{TRS-old-tree} with \figref{TRS-new-tree}, we can
see that the divide-and-conquer tree of original algorithm is
not changed, the new ND algorithm just replaces some original
solid arrows of ``$\serial$'' primitive by dashed
arrows of ``$\fire$'' operator to hint the runtime scheduler
that some subtrees can possibly be scheduled earlier as soon
as their data dependencies are satisfied.
Since the structure of spawn tree is preserved, we argue that
if the runtime system schedule any local computation
in a depth-first order with respect to the spawn tree,
the cache locality should not suffer.

} 

\begin{figure}[h]
\includegraphics[width = 0.7 \linewidth]{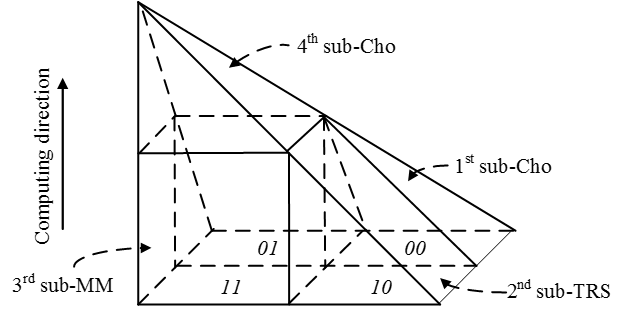}
\caption{$2$-way divide-and-conquer Cholesky algorithm}
\label{fig:cho-geo-algo}
\end{figure}

\para{Cholesky Decomposition.}
Given an $n$(row)-by-$n$(column) Hermitian, positive-definite
matrix $A$, the Cholesky decomposition asks for an $n$-by-$n$
lower triangular matrix $L$ such that $A = L L^T$. We denote 
the algorithm as $L \assign \proc{Cho}(A)$. This problem can
be recursively solved by a $2$-way divide-and-conquer algorithm,
geometrically described in \figref{cho-geo-algo}, as follows:
\begin{align*}
    \begin{bmatrix}
        A_{00} & A_{10}^T \\
        A_{10} & A_{11}
    \end{bmatrix}
    &=
    \begin{bmatrix}
        L_{00} & 0 \\
        L_{10} & L_{11}
    \end{bmatrix}
    \begin{bmatrix}
        L_{00}^T & L_{10}^T \\
        0        & L_{11}^Tc
    \end{bmatrix} \nonumber \\
    &=
    \begin{bmatrix}
        L_{00} L_{00}^T & L_{00} L_{10}^T \\
        L_{10} L_{00}^T & L_{10} L_{10}^T + L_{11} L_{11}^T
    \end{bmatrix}.
\end{align*}

Cholesky factorization can be expressed in the fork-join model as shown in
\eqref{Cho-old-sched} with a span recurrence of $T_{\infty, \proc{Cho}}(n) = 
2 T_{\infty, \proc{Cho}}(n/2) + T_{\infty, \proc{TRS}}(n/2) + T_{\infty, \proc{MM}}(n/2)$.
Assuming spans of $O(n \log n)$ and $O(n)$  for  \proc{TRS} and \proc{MM} respectively,
this recurrence results in a span bound of $O(n \log^2 n)$ for the $2$-way 
divide-and-conquer Cholesky algorithm.
\begin{align}
L &\assign \proc{Cho}(A) = \nonumber \\
  & ( L_{00} \assign \proc{Cho}(A_{00})
    \serial L_{10} \assign \proc{TRS}(L_{00}, A_{10}^{T})^{T} )\nonumber \\
    \serial& ( \proc{MMS}(L_{10}, L_{10}^T, A_{11}) \serial L_{11} \assign \proc{Cho}(A_{11}))
\label{eq:Cho-old-sched}
\end{align}

We can express Cholesky in the ND model as follows:
\begin{align}
L &\assign \proc{Cho}(A) = \nonumber \\
  & ( L_{00} \assign \proc{Cho}(A_{00})
    \tfire{CT} L_{10} \assign \proc{TRS}(L_{00}, A_{10}^{T})^{T} ) \nonumber \\
    \tfire{CTMC} & ( \proc{MMS}(L_{10}, L_{10}^T, A_{11})
    \tfire{MC} L_{11} \assign \proc{Cho}(\tilde{A}_{11})).
\label{eq:Cho-new-sched}
\end{align}

The set of fire rules are defined as follows 
(note that ``$\tfire{TM}$'' is the same as in \eqref{TM}):

\begin{align*}
    \oone{+} \tfire{CT} \oone{-} = \{ &
       \othree{+}{1}{1} \tfire{CT} \ofour{-}{1}{1}{1}, 
       \othree{+}{1}{1} \tfire{CT} \ofour{-}{1}{2}{1}, \\
     & \othree{+}{1}{2} \tfire{TM2} \ofour{-}{1}{2}{1}, 
       \othree{+}{1}{2} \tfire{TM2} \ofour{-}{1}{2}{2}, \\
     & \othree{+}{2}{2} \tfire{CT} \othree{-}{2}{1}, 
       \othree{+}{2}{2} \tfire{CT} \othree{-}{2}{2} 
     \} \\
    \oone{+} \tfire{CTMC} \oone{-} = \{ & \otwo{+}{2} \tfire{TM2} \otwo{-}{1} \} \\
	\oone{+} \tfire{TM2} \oone{-} = \{ &
		\oone{+} \tfire{TM} \oone{-},
		\oone{+} \tfire{TM1} \oone{-}
	\}	\\
	\oone{+} \tfire{TM1} \oone{-} = \{ &
		  \ofour{+}{1}{1}{1} \tfire{TM1} \ofour{-}{1}{1}{1},  
		  \ofour{+}{1}{1}{1} \tfire{TM1} \ofour{-}{1}{1}{2}, \\ 
		& \ofour{+}{1}{2}{1} \tfire{TM1} \ofour{-}{1}{1}{1},  
		  \ofour{+}{1}{2}{1} \tfire{TM1} \ofour{-}{1}{1}{2}, \\ 
		& \othree{+}{2}{1} \tfire{TM1} \ofour{-}{2}{1}{1},  
		  \othree{+}{2}{1} \tfire{TM1} \ofour{-}{2}{1}{2}, \\ 
		& \othree{+}{2}{2} \tfire{TM1} \ofour{-}{2}{2}{1},  
		  \othree{+}{2}{2} \tfire{TM1} \ofour{-}{2}{2}{1}  
	\}	\\
    \oone{+} \tfire{MC} \oone{-} = \{ &
       \ofour{+}{2}{1}{1} \tfire{MC} \othree{-}{1}{1}, 
        \ofour{+}{2}{2}{1} \tfire{MT} \othree{-}{1}{2}, \\
       & \ofour{+}{2}{2}{2} \tfire{MC} \othree{-}{2}{2} 
     \}
\end{align*}
The span recurrence for Cholesky is:
\begin{align*}
& T_{\infty, \proc{Cho}}(n) 
&\le T_{\infty, \proc{Cho}\otwo{1}{1} \tfire{CT} \otwo{1}{2} \tfire{TM} \otwo{2}{1}}(n) 
+ T_{\infty, \proc{Cho}\otwo{2}{2}}(n) \nonumber \\
&& = T_{\infty, \proc{Cho}\otwo{1}{1} \tfire{CT} \otwo{1}{2} \tfire{TM} \otwo{2}{1}}(n) 
+ T_{\infty, \proc{Cho}}(n/2) \nonumber \\
\end{align*}
The first term on the right hand side can be bounded recursively by
\begin{equation}
 T_{\infty, \proc{Cho}\otwo{1}{1} \tfire{CT} \otwo{1}{2} \tfire{TM} \otwo{2}{1}}(n) 
\le 2 T_{\infty, \proc{Cho}\otwo{1}{1} \tfire{CT} \otwo{1}{2} \tfire{TM} \otwo{2}{1}}(n/2).
\label{eq:Cho-new-span-rec}
\end{equation}
For the base case, we have 
\[
 T_{\infty, \proc{Cho} \tfire{CT} \proc{TRS} \tfire{TM} \proc{MM}} (1)
= T_{\infty, \proc{Cho}}(1) + T_{\infty, \proc{TRS}}(1) + T_{\infty, \proc{MM}}(1)
\]
\eqref{Cho-new-span-rec} solves to $O(n)$ and is asymptotically optimal.

\vspace{1ex}
\para{LU with Partial Pivoting.}
A straightforward parallelization of the $2$-way divide-and-conquer
algorithm by Toledo \cite{Toledo97}, combined with a replacement of
the \proc{TRS} algorithm by our new ND \proc{TRS}, yields an optimal
LU with partial pivoting algorithm for an $n\times m$ matrix with time
(span) bound $O(m \log n)$, and serial cache bound $O(nm^2/B\sqrt{M} +
n m + (n m \log m)/B)$ in the ideal cache model~\cite{FrigoLePr99}
with a cache of size $M$ and block size $B$.

\vspace{1ex}
\para{Floyd-Warshall Algorithm.}
\begin{figure}[!ht]
\hspace{-.25in}
\includegraphics[width=0.45\linewidth]{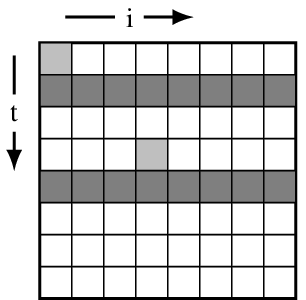}
\vspace{5pt}
\caption{1D FW dependency pattern}
\label{fig:floyd-1D-def}
\end{figure}

The fire construct can also be used to express dynamic programming
algorithms.  We will demonstrate this with 1-dimensional
Floyd-Warshall, a simple synthetic benchmark originally introduced
in \cite{TangYoKa15} \punt{to illustrate the ``eager'' COW
(Cache-Oblivious Wavefront) technique.}. Its data dependency pattern
is similar to that of the Floyd-Warshall algorithm for All-Pairs
Shortest Paths.  The defining recurrence of $1$D FW is is as follows
for $1 \leq i, t \leq n$ (we assume that $d(0, i)$ are already known
for $1 \leq i \leq n$):

\vspace{-6pt}
\begin{align}
    d(t, i) &= d(t - 1, i) \oplus d(t - 1, t - 1).
\label{eq:floyd-1D-def}
\end{align}

\figref{floyd-1D-def} shows the data dependency pattern of $1$D FW. In the figure,
dark-shaded cells are those updated in the current timestep, and the
light-shaded cells denote the diagonal cells from previous time step
used to calculate the current row. The value of cell $i$ at timestep
$t$, $d(t, i)$ depends on the value of the cell at the previous
timestep, $d(t - 1, i)$, and the value of the diagonal cell from the
previous timestep, $d(t-1, t-1)$.

We adapt the divide-and-conquer algorithm from \cite{ChowdhuryRa10} to
recursively solve the problem as follows: given a dynamic programming
table $X$, we apply the following algorithm, $A$, to $X$:
\begin{align}
X \assign \proc{A}(X) =
  & ((X_{00} \assign \proc{A}(X_{00})) \tfire{AB}
  (X_{01} \assign \proc{B}(X_{01}, X_{00}))) \nonumber \\
\tfire{ABAB} & \ ((X_{11} \assign \proc{A}(X_{11})) \tfire{AB} 
  (X_{10} \assign \proc{B}(X_{10}, X_{11}))) \nonumber \\
X \assign \proc{B}(X, Y) =
  & ((X_{00} \assign \proc{B}(X_{00}, Y_{00})) \parallel 
  (X_{01} \assign \proc{B}(X_{01}, Y_{00}))) \nonumber \\
\tfire{BBBB} & \ ((X_{10} \assign \proc{B}(X_{10}, Y_{11})) \parallel
  (X_{11} \assign \proc{B}(X_{11}, Y_{11})))
\label{eq:fw-1D-new-sched}
\end{align}
In \eqref{fw-1D-new-sched}, $X \assign \proc{A}(X)$ denotes a task on 
data block $X$ that contains all the diagonal entries needed for the task,
and $X \assign \proc{B}(X, Y)$  denotes the task on data block $X$ where the
diagonal entries needed for the task are contained in $Y$.

The set of fire rules is as follows:
\begin{flalign*}
    & \begin{aligned}[t]\oone{+} \tfire{AB} \oone{-} = \{
     & \othree{+}{1}{1} \tfire{AB} \othree{-}{1}{1},
     \othree{+}{1}{1} \tfire{AB} \othree{-}{1}{2}, \\
     & \othree{+}{2}{1} \tfire{AB} \othree{-}{2}{1},
     \othree{+}{2}{1} \tfire{AB} \othree{-}{2}{2}   \}
    \end{aligned}\\
    & \oone{+} \tfire{ABAB} \oone{-} =
    \{ \otwo{+}{2} \tfire{BA} \otwo{-}{1} \}\\ 
    & \oone{+} \tfire{BA} \oone{-} = 
    \{ \othree{+}{2}{1} \tfire{BA} \othree{-}{1}{1},
    \othree{+}{2}{2} \tfire{BB} \othree{-}{1}{2} \}\\
    & \oone{+} \tfire{BBBB} \oone{-} =
    \{ \otwo{+}{1} \tfire{BB} \otwo{-}{1},
    \otwo{+}{2} \tfire{BB} \otwo{-}{2} \}\\
    & \oone{+} \tfire{BB} \oone{-} = \{ \othree{+}{2}{1} \tfire{BB} \othree{-}{1}{1}, \othree{+}{2}{2} \tfire{BB} \otwo{1}{2} \}
\end{flalign*}

\punt{
Of course, the computation of $X_{11} \assign
\proc{A}(X_{11})$ depends not only on the diagonal cells contained in $X_{11}$
but also on the last row of $X_{01}$, which is computed by $\proc{B}\oone{2}$.
However, since this data dependency is automatically satisfied by the ``$\serial$``
primitive in the schedule, we omit $X_{01}$ as an input parameter to $\proc{A}\oone{3}$
for simplicity. Similarly we omit the first row of $X$, i.e. $d(0, i)$ for
$1 \leq i \leq n$, as the input to $\proc{A}\oone{1}$, and so on.
}

If ``$\tfire{ABAB}$'' and ``$\tfire{BBBB}$'' are regarded as
``$\serial$'' (which only increases the span), the recurrence for span
in the ND model is:
 \begin{equation}
 	T_{\infty, \proc{A}}(n)  \le 2 T_{\infty, \tfire{AB}}(n/2), \quad \quad
 	T_{\infty, \tfire{AB}}(n) \le 2 T_{\infty, \tfire{AB}}(n/2).
 \label{eq:fw-1D-new-span-rec}
 \end{equation}
With the base case $T_{\infty, \tfire{AB}}(1) =
O(1)$, \eqref{fw-1D-new-span-rec} solves to the optimal $O(n)$, as
opposed to $O(n \log n)$ in the NP model.

Expressing the original $2$D ($2$ spatial dimensions plus $1$ time
dimension) Floyd-Warshall all-pairs-shorest-paths \cite{Warshall62,
Floyd62} using the ``$\fire$'' construct is a straightforward
extension of the design demonstrated here.

\punt{\subsecput{dp}{Dynamic Programming}
\para{Floyd-Warshall style Dynamic Programming}
We start from the simple synthetic $1$D ($1$D space) FW. 
It bears similarity to original $2$D FW \cite{Floyd62} in the 
data dependency pattern.  The defining recurrence of $1$D FW 
for $1 \leq i, t \leq n$ is as follows which assumes that 
$d(0, i)$ for $1 \leq i \leq n$ are already known.
\begin{align}
    d(t, i) &= d(t - 1, i) \oplus d(t - 1, t - 1)
\label{eq:floyd-1D-def}
\end{align}

We can see that the update of any cell $d(t, i)$ in $1$D FW depends on
both the diagonal cell from previous time step, i.e. $d(t-1, t-1)$,
and the cell of the same space position from previous time step,
i.e. $d(t-1, i)$.

We adapted a $2$ way divide-and-conquer algorithm from
\cite{ChowdhuryRa10} to recursively solve the problem as follows. 
Since the entire computation space is $2$D, i.e. one space dimension
plus one time dimension, at each recursion level we divide the
computation space $X$ into four equally sized sub-quadrants $X_{00}$,
$X_{01}$, $X_{10}$, and $X_{11}$, and apply schedule
in \eqref{fw-1D-new-sched} to the divide-and-conquer tree.
\noindent
\begin{figure}[!ht]
\begin{align}
X &\assign \proc{A}(X) = \nonumber \\
  & \Comment (\proc{A}\oone{1} \fire \proc{B}\oone{2}) \nonumber \\
  & (X_{00} \assign \proc{A}(X_{00}) \fire X_{01} \assign \proc{B}(X_{01}, X_{00})) \nonumber \\
  & \Comment \serial (\proc{A}\oone{3} \fire \proc{B}\oone{4}) \nonumber \\
  & \serial (X_{11} \assign \proc{A}(X_{11}) \fire X_{10} \assign \proc{B}(X_{10}, X_{11})) \nonumber \\
X &\assign \proc{B}(X, Y) = \nonumber \\
  & \Comment (\proc{B}\oone{1} \parallel \proc{B}\oone{2}) \nonumber \\
  & (X_{00} \assign \proc{B}(X_{00}, Y_{00}) \parallel X_{01} \assign \proc{B}(X_{01}, Y_{00})) \nonumber \\
  & \Comment \serial (\proc{B}\oone{3} \parallel \proc{B}\oone{4}) \nonumber \\
  & (X_{10} \assign \proc{B}(X_{10}, Y_{11}) \parallel X_{11} \assign \proc{B}(X_{11}, Y_{11}))
\label{eq:fw-1D-new-sched}
\end{align}
\end{figure}
In \eqref{fw-1D-new-sched}, $X \assign \proc{A}(X)$ denotes a computation on
data block $X$ with all dependent diagonal cells self-contained. 
\punt{
Of course, the computation of $X_{11} \assign 
\proc{A}(X_{11})$ depends not only on the diagonal cells contained in $X_{11}$
but also on the last row of $X_{01}$, which is computed by $\proc{B}\oone{2}$.
However, since this data dependency is automatically satisfied by the ``$\serial$``
primitive in the schedule, we omit $X_{01}$ as an input parameter to $\proc{A}\oone{3}$
for simplicity. Similarly we omit the first row of $X$, i.e. $d(0, i)$ for 
$1 \leq i \leq n$, as the input to $\proc{A}\oone{1}$, and so on.
} 
$X \assign \proc{B}(X, Y)$ denotes a computation on data block $X$
with all dependent diagonal cells included in a completely disjoint
data block $Y$.  In \eqref{fw-1D-new-sched}, there is only one type of
``$\proc{A} \fire \proc{B}$'' that needs to be refined. In order to
know how to expand the ``$\fire$'' recursion, we expand the recursion of
\proc{A}\oone{x} and \proc{B}\oone{y} respectively for one more level and
match the input and output of sub-function calls with each. We have the
``$\fire$'' rule in \eqref{fw-1D-fire}.
\noindent
\begin{figure}[!ht]
\begin{align}
\proc{A}\oone{x} &\fire \proc{B}\oone{y} = \nonumber \\
	& (\proc{A}\otwo{x}{1} \fire \{\proc{B}\otwo{x}{2}, \proc{B}\otwo{y}{1}, \proc{B}\otwo{y}{2}\}) \nonumber \\
	& \serial (\proc{A}\otwo{x}{3} \fire \{\proc{B}\otwo{x}{4}, \proc{B}\otwo{y}{3}, \proc{B}\otwo{y}{4}\})
\label{eq:fw-1D-fire}
\end{align}
\end{figure}
A corresponding span recurrence is in \eqref{fw-1D-new-span-rec}, and
solves to $O(n)$, which is asymptotically optimal.
\noindent
\begin{figure}[!ht]
\begin{align}
	T_{\infty, \proc{A}}(n) &= 2 T_{\infty, \proc{A} \fire \proc{B}}(n/2) \nonumber \\
	T_{\infty, \proc{A} \fire \proc{B}}(n) &= 2 T_{\infty, \proc{A} \fire \proc{B}}(n/2) \nonumber \\
	T_{\infty, \proc{A} \fire \proc{B}}(1) &= T_{\infty, \proc{A}}(1) + T_{\infty, \proc{B}}(1) = O(1) 
\label{eq:fw-1D-new-span-rec}
\end{align}
\end{figure}

Applying ``$\fire$'' operator to original $2$D ($2$ spatial dimensions
plus $1$ time dimension) Floyd-Warshall All-Pairs-Shorest-Paths
\cite{Warshall62, Floyd62} follows exactly the same procedure, just more
complicated.  So we leave it to interested readers.}

\begin{figure*}[!ht]
\begin{minipage}[b]{0.24 \linewidth}
\subfloat[Dashed arrows defined by the algorithm in \eqref{lcs-ND} and
rerwiting rules \eqreftwo{lcs-hv-rule}{lcs-vh-rule}.]{
\label{fig:lcs-dag-init}
\includegraphics[width=0.95\linewidth]{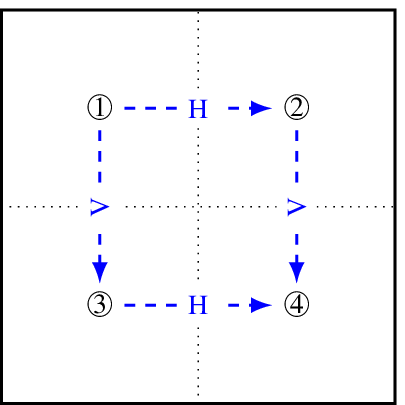}
} 
\end{minipage}
\hfill
\begin{minipage}[b]{0.24 \linewidth}
\subfloat[Dashed arrows are rewritten by rules in \eqreftwo{lcs-h-rule}{lcs-v-rule}.]{
\label{fig:lcs-dag-rewrite}
\includegraphics[width=0.95\linewidth]{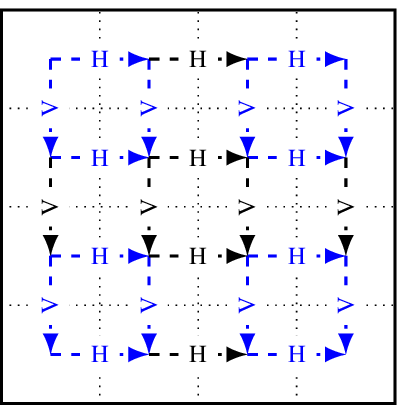}
} 
\end{minipage}
\hfill
\begin{minipage}[b]{.5\linewidth}
\subfloat[Spawn tree of LCS in ND model. We only draw one
 ``$\fire$'' path in \figref{lcs-dag-rewrite} from top-left
 to bottom-right cell]{
\label{fig:lcs-new-tree}
\includegraphics[width=0.95\linewidth]{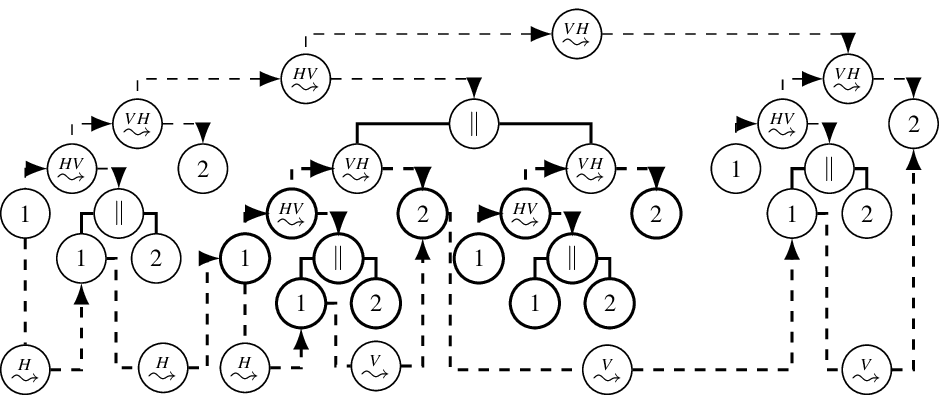}
} 
\end{minipage}
\caption{DAG Rewriting and spawn tree of LCS in ND model}
\label{fig:lcs-nd-model}
\vspace{-15pt}
\end{figure*}

\vspace{2ex}
\para{LCS (Longest Common Subsequence).}
In this section, we express a divide-and-conquer algorithm for LCS in
ND model.  Given two sequences $S = \langle \allowbreak
s_1, \allowbreak s_2, \allowbreak \ldots, \allowbreak s_m \rangle$ and
$T = \langle \allowbreak t_1, \allowbreak
t_2, \allowbreak \ldots, \allowbreak t_n \rangle$, the goal is to find
the length of longest common subsequence of $S$ and $T$. LCS can be
computed using \eqref{LCS-def} \cite{CormenLeRi09}.
\footnote{A similar recurrence applies to the pairwise sequence
 alignment with affine gap cost \cite{Gotoh82}.}
\noindent
\begin{align}
    X(i, j) &= 
        \cases{ 
        0 &  \mbox{if } i = 0 ~\vee~ j = 0\\
        X(i-1, j-1) + 1 &  \mbox{if } i, j > 0 ~\wedge~ s_i = t_j\\
        \max\{X(i, j-1), X(i-1, j)\} &  \mbox{if } i, j > 0 ~\wedge~ s_i \neq t_j
        } 
\label{eq:LCS-def}
\end{align}
In the ND model, we express the divide and conquer algorithm for LCS
that solves the above recursion for two sequences of the same length
($n$) as follows (see~\figref{lcs-new-tree}):
\noindent
\begin{flalign}
X \assign \proc{LCS}(X) = &
   ( (X_{00} \assign \proc{LCS}(X_{00}))  \tfire{\proc{HV}}  \nonumber \\
 & \quad (X_{01} \assign \proc{LCS}(X_{01}) 
  \parallel X_{10} \assign \proc{LCS}(X_{10})))  \nonumber \\
   \tfire{\proc{VH}}
   & \  (X_{11} \assign \proc{LCS}(X_{11}) )
\label{eq:lcs-ND}
\end{flalign}

The partial dependencies are given by the following fire rules which
are illustrated in~\figreftwo{lcs-dag-init}{lcs-dag-rewrite}:
\noindent
\begin{equation}
\oone{+} \tfire{\proc{HV}} \oone{-} = \{\oone{+} \tfire{\proc{H}} \otwo{-}{1},
\oone{+} \tfire{\proc{V}} \otwo{-}{2}\} 
\label{eq:lcs-hv-rule}
\end{equation}
\begin{equation}
\oone{+} \tfire{\proc{VH}} \oone{-} = \{\otwo{+}{1} \tfire{\proc{V}} \oone{-},
\otwo{+}{2} \tfire{\proc{H}} \oone{-}\} 
\label{eq:lcs-vh-rule}
\end{equation}
\begin{equation}
\oone{+} \tfire{\proc{H}} \oone{-} = \{\ofour{+}{1}{2}{1} \tfire{\proc{H}} \othree{-}{1}{1},
\otwo{+}{2} \tfire{\proc{H}} \ofour{-}{1}{2}{2}\}  
\label{eq:lcs-h-rule}
\end{equation}
\begin{equation}
\oone{+} \tfire{\proc{V}} \oone{-} = \{\ofour{+}{1}{2}{2} \tfire{\proc{V}} \othree{-}{1}{1}, 
\otwo{+}{2} \tfire{\proc{V}} \ofour{-}{1}{2}{1}\}  
\label{eq:lcs-v-rule}
\end{equation}

To compute the span of LCS, consider the dynamic programming
table. The span is defined by the length of longest path in the DAG
which runs from the top left entry to the bottom right entry. We will
separately compute the length of the longest horizontal path,
$T_{h}(n)$, and the length of the longest vertical path,
$T_{v}(n)$. Notice that the span, $T_{\infty, LCS}(n)$, is bounded
above by $T_{h}(n) + T_{v}(n)$.

Since we split an LCS problem whose dynamic programming table is of
size $n\times n$ into four LCS problems of size $n/2\times n/2$ of
which the longest horizontal path covers two, we have
$T_{h}(n)=2T_{h}(n/2)$. The base case (a $1\times 1$ matrix) only
depends on three inputs, so that $T_{h}(1)=O(1)$.  Therefore,
$T_{h}(n)=O(n)$. Similar reasoning shows that $T_{v}(n)=O(n)$.  As a
result, $T_{\infty, LCS}(n)$ is bounded above by $O(n)$, which is
optimal.

\secput{sb-sched}{Space-Bounded Schedulers for the ND Model}

We show that {\it reasonably regular} programs in the ND model,
including all the algorithms in~\secref{num}, can be effectively
mapped to Parallel Memory Hierarchies by adapting the design of
space-bounded (SB) schedulers for NP programs. Regularity is a
quantifiable property of the algorithm (or spawn tree) that measures
how difficult it is to schedule; we will quantify this and argue show
that the algorithm in~\secref{num} are highly regular.  Space-bounded
schedulers for programs in the NP model were first proposed for
completely regular programs \cite{ChowdhurySiBl13}, improved upon and
rigorously analyzed in \cite{BlellochFiGi11}, and empirically
demonstrated to outperform work-stealing based schedulers for many
algorithms in \cite{SimhadriBlFi14}, but not for TRSM and Cholesky
algorithms due to their limited parallelism in the NP
model~\cite{SimhadriBlFi15}. The key idea in SB schedulers is that
each task is annotated with the size of its memory footprint to guide
the mapping of tasks to processors and caches in the hierarchy.  The
main result of this section is~\thmref{rt}, which says that the SB
scheduler is able to exploit the extra parallelism exposed in the ND
model.

\para{Machine Model: Parallel Memory Hierarchy.}
SB schedulers are well suited for the Parallel Memory Hierarchy (PMH)
machine model \cite{ACF93} (see~\figref{PMH}), which models the
multi-level cache hierarchies and cache sharing common in shared
memory multi-core architectures.  The PMH is represented by a
symmetric tree rooted at a main memory of infinite size. The internal
nodes are caches and the leaves are processors. We refer to subtrees
rooted at some cache as \defn{subclusters}. Each cache at level $i$ is
assumed to be of the same size $M_i$, and has the same the number of
level-$(i-1)$ caches attached to it. We call this the \defn{fan-out}
of level-$i$ and denote it by the constant $f_i$, so that the number
of processors in a $h$-level tree if $p_h = \prod_{i=1}^{h} f_i$. We
let $M_0$ denote a constant indicative of the number of registers on a
processor. We let $C_{i-1}$ denote the cost parameter representing the
cost of servicing a cache miss at level $(i-1)$ from level $i$. A
cache miss that must be serviced from level $j$ requires $C'_j =
C_0+C_1+\dots+C_{j-1}$ time steps.  For simplicity, we let the cache
block be one word long.  This limitation can be relaxed and analyzed
as in \cite{BlellochFiGi11}.

\para{Terminology.} 
A task is \defn{done} when all the leaf nodes (strands) associated with
its subtree have been executed.  A dataflow arrow originating at a leaf
node in the spawn tree is \defn{satisfied} when its source node is done.
A dataflow arrow originating at an internal node of the spawn tree is
\defn{satisfied} when all its descendants (rewritings) according to the
fire rules have been satisfied.  
A task is \defn{fully ready} or just \defn{ready} when all the
incoming dependencies (dataflow arrows) originating outside the
subtree are satisfied.  The \defn{size}, $\Sl{\cdot}$, of a task or a
strand is the number of distinct memory locations accessed by it.  We
assume that programs are statically allocated, that is all necessary
heap space is allocated up front and freed at program termination, so
that the size function is well defined. The size annotation can be
supplied by the programmer or can be obtained from a profiling
tool. If the size of a task in the spawn tree is not specified, we
inherit the annotation from its lowest annotated ancestor in the spawn
tree. We call a task \defn{$M$-maximal} if its size is at most $M$,
but its parent in the spawn tree has size $>M$. A task is level-$i$
maximal in a PMH if it is $M_i$-maximal, $M_i$ being the size of a
level-$i$ cache. Note that even though an $M_i$-maximal task is not
ready, a $M_j$-maximal subtask inside it (where $j<i$) can be ready.

\vspace{1ex}
\para{SB Schedulers.}
We define a space-bounded scheduler to be any scheduler that has the
anchoring and boundedness properties \cite{SimhadriBlFi15}:
\begin{itemize}
\item[\textbf{Anchor:}]
As the spawn tree unfolds dynamically, we assign and \defn{anchor}
ready tasks to caches in the hierarchy with respect to which they are
maximal.  Tasks are \textit{allocated} a part of the subcluster rooted
at the assigned cache. The anchoring property requires that all the
leaves of the spawn tree of a task be executed by processors in the
part of the subcluster allocated.
\item[\textbf{Boundedness:}]
 Tasks anchored to a cache of size $M$ have a total size $\leq\sigma M$,
where $\sigma\in(0,1)$ is a scheduler chosen \defn{dilation parameter}.
\end{itemize}

There are several ways to maintain these properties and operate within
its constraints. The approach taken in \cite{BlellochFiGi11} is to
have a task queue with each anchored task that contains its subtasks
than can be potentially unrolled and anchored to the caches below it.
We adopt the same approach here (outlined below for convenience) for
the ND model with the \textbf{difference} being that we only anchor
and run ready subtasks. In the course of execution, ready tasks are
anchored to a suitable cache level (provided there is sufficient space
left), and each anchored task is \defn{allocated} subclusters beneath
the cache, based on the size of the task. Just as
in \cite{BlellochFiGi11}, a task of size $S$ anchored at level-$i$
cache is allocated
\[
g_i(S) = \min\{f_i,                                         
\max\{1,\floor{f_i(3S/\MM_i)^{\alpha'}}\} \}, \textrm{where}\  
\alpha' = \min\{\alpha_{max},1\}
\] 
level-$(i-1)$ subclusters \footnote{The factor 3 in the allocation
function is a detail necessary to prove Thm.\ref{thm:rt}.}
where $\alpha_{max}$ is the parallelizability
of the task, a term we will define shortly. All processors in the
subclusters are required to work exclusively on this task.  Initially, 
the root node of the spawn tree is anchored to the root of the PMH.

To find work, a processor traverses the path from the leaf it
represents in the tree towards the root of the PMH until it reaches
the lowest anchor it is part of. Here it checks for ready tasks in the
queue associated with this anchor, and if empty, re-attempts to find
work after a short wait.  Otherwise, it pulls out a task from the
queue.  If the task is from an anchor at the cache immediately above
the processor, i.e. {\bf at an $L_1$ cache}, it executes the subtask
by traversing the corresponding spawn tree in depth-first order.  If
the processor pulled this task out of an allocation at a {\bf cache at
level $i>1$}, it does not try to execute its strands (leaves)
immediately. Instead, it unrolls the spawn tree corresponding to the
task using the DRS and enqueues those subtasks that are either of size
$>M_{i-1}$, or not ready, in the queue corresponding to the
anchor. Those subtasks that cannot be immediately worked on due to
lack of space in the caches are also enqueued. However, if the
processor encounters a ready task that has size less than that of a
level-$j$ cache ($j<i$), and is able to find sufficient space for it
in the subcluster allocated to the anchor, the task is anchored at the
level-$j$ cache, and allocated a suitable number of subclusters below
the level-$j$ cache.  The processor starts unraveling the spawn tree
and finding work repeatedly.  When an anchored task is done, the
anchor, allocation and the associated resources are released for
future tasks. We also borrow other details in the design of the
space-bounded schedulers (e.g. how many subclusters are provisioned
for making progress on ``worst case allocations''? what fraction of
cache is reserved for tasks that ``skip cache levels''?)  from prior
work \cite{BlellochFiGi11}.

Roughly speaking, this scheduler uses all the partial parallelism
between level-$(i-1)$ maximal subtasks within a level-$i$ maximal
task. However, it does not use all the partial parallelism across
level-$(i-2)$ subtasks, especially those dataflow arrows between
level-$(i-2)$ subtasks in two different level-$(i-1)$ subtasks
(see~\figref{succ-levels}).

\begin{figure}[!th]
\includegraphics[width=\linewidth]{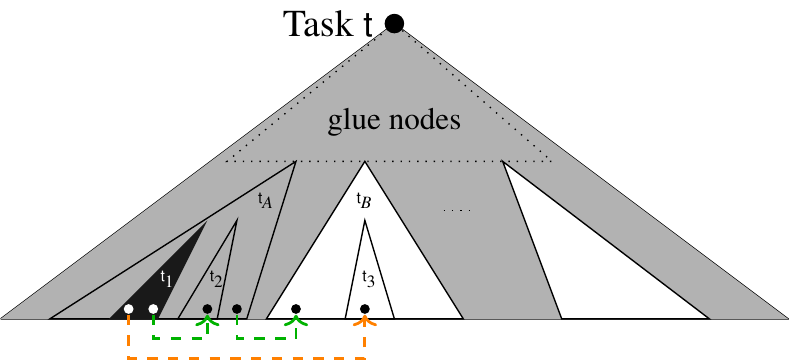}
\vspace{1ex}
\caption{Use of partial parallelism in the SB scheduler.  In this
  diagram, white represents tasks that are yet to start, gray
  represents running tasks, and black represents complete tasks.
  Green arrows represent dataflow arrows that may be used to start new
  tasks by the SB scheduler while orange datalow arrows are never
  immediately used.  Task $\Qt$ is level-$i$ maximal; tasks $\Qt_A$
  and $\Qt_B$ are level-$(i-1)$ maximal; tasks $\Qt_1,\Qt_2,$ and
  $\Qt_3$ are level-$(i-2)$ maximal.  Although subtask $\Qt_1$ has
  completed and has two outgoing dataflow edges, only $\Qt_2$, which
  is in the same level-$(i-1)$ maximal subtask ($\Qt_A$) can be
  started; $\Qt_3$ can not immediately started until subtask $\Qt_A$
  completes.
\label{fig:succ-levels}}
\end{figure}


\para{Metrics.} We now analyze the running time of the SB scheduler,
accounting for the cost of executing the work and load imbalance, but
not the overhead of the data structures need to keep track of anchors,
allocations, and the readiness of subtasks. We leave the optimization
of this overhead for a future empirical study.  The anchoring and
boundedness properties make it easy to preserve locality while trading
off some parallelism. Inspired by the analysis in
\cite{BlellochFiGi11,Simhadri13}, we develop a new analysis for the
ND model to argue that the impact of the loss of parallelism caused by
the anchoring property on load balance is not significant.

A critical consequence of the anchoring property of the SB scheduler
is that once a task is anchored to a cache, all the memory locations
needed for the task are loaded only once and are not forced to be
evicted until the completion of the task. This motivates the following
quantification of locality.  Given a task $\Qt$, decompose the spawn
tree into $M$-maximal subtasks, and ``glue nodes'' that hold these
trees together (this decomposition is unique).  Define the
\defn{parallel cache complexity (PCC)}, $\cc(\Qt;M)$, of task $\Qt$ 
to be the sum of sizes of the maximal subtrees, plus a constant
overhead from each glue node. This is motivated by the expectation
that a good scheduler (such as SB) should be able to preserve locality
within $M$-sized tasks given an cache of appropriate size, while it
might be too cumbersome to preserve locality across maximal
subtasks. \footnote{This definition is a generalization
of \cite[Defn.2]{BlellochFiGi11} for the ND model.  The full metric
measures cache complexity in terms of cache lines to model latency and
is also parameterized by a second parameter $B$: size of a cache
line. We set $B=1$ here for simplicity.  This simplification can be
reversed.}  The PCC metric differs from the another common metric for
locality of NP programs: the cache complexity $Q_1$ of the depth-first
traversal in the ideal cache model~\cite{AcarBlBl00}. Unlike $Q_1$,
$\cc$ does not depend on the order of traversal, but does not capture
data reuse across $M$-maximal subtasks, which is a smaller order term
in our algorithms.

Note that $M$ is a free parameter in this analysis. When the context
is clear, we often replace the task $\Qt$ in the $\cc$ expression with
a size parameter corresponding to the task, so that cache complexity
is denoted $\cc(N;M)$. With this notation we have the following bound
on the cache complexity of the algorithms in ~\secref{num}.

\begin{claim}
For dense matrices of size $N = n\times n$, the divide and conquer
classical matrix multiplication, Triangular System Solve, Cholesky
and LU factorizations, and the 2D analog of the Floyd-Warshall algorithm
in~\secref{num} have parallel cache complexity 
\[\cc(N;M) = O( N^{1.5}/M^{0.5}),\]
when $N>M$, with the glue nodes contributing an asymptotically
smaller term. The LCS algorithm has $\cc(n;M) = O(n^{2}/M)$ for input
of size $n>\sqrt{M}$. This is true even if the algorithms are
expressed in the NP model by replacing fire constructs with ``$\serial$''.
\label{clm:pcc-algos}
\end{claim}

As a direct consequence of the anchoring and boundedness properties,
which conservatively provision cache space, the following restatement
of \cite[Theorem 3] {BlellochFiGi11} applies to the ND model with the
same proof.
\begin{theorem}
\vspace{-3pt}
Suppose $\Qt$ is a task in ND program that is anchored at a level-$i$
cache of a PMH by a SB scheduler with dilation parameter $0<\sigma<1$
(i.e., a SB scheduler that anchors tasks of size at most $\sigma M_j$
at level $j$). Then for all cache levels $j\leq i$, the sum of
cache misses incurred by all caches at level $j$ is at most
$\cc(\Qt;\sigma\cdot M_j)$.
\end{theorem}

In conjunction with~\clmref{pcc-algos}, this gives a bound on
the communication cost of the schedulers for ND algorithms. One can
verify from results on lower bounds on communication
complexity~\cite{BallardDeHoSc11} that these bounds are asymptotically
optimal.  If the scheduler is able to perfectly load balance a program
at every cache level on an $h$ level PMH with $p$ processors, we would
expect a task $\Qt$ to take
\begin{equation}
\frac{\sum_{i=0}^{h-1} \cc(\Qt;\sigma\cdot M_i)\cdot C_i}{p}
\label{eq:lb}
\end{equation}
time steps to complete, where $0<\sigma<1$ is the dilation parameter.

\begin{figure}[!t]
\includegraphics[width=\linewidth]{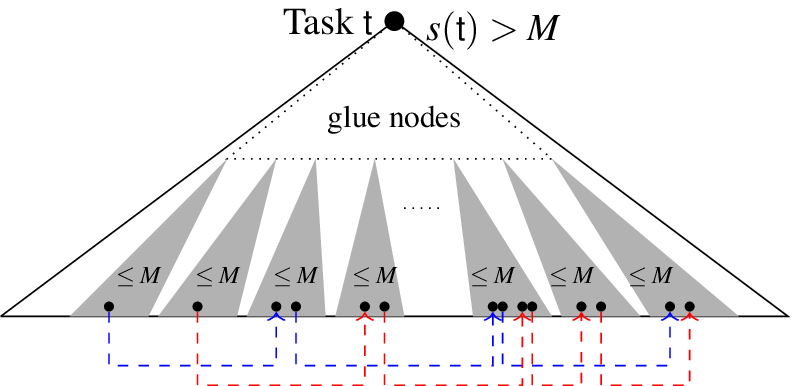}
\vspace{1ex}
\caption{$M$-maximal subtasks (in gray) and glue nodes in the spawn tree 
of a task $\Qt$. The PCC, $\cc(\Qt;M)$, is the sum of sizes of $M$-maximal
subtasks plus one miss for each glue node. The red and blue sets of arrows
represent two chains of dependencies in $\Qt$. The ECC, $\bcc{\Qt;M},$ is
determined by the maximum, among all such chains, of the sum of effective
depth of $M$-maximal subtasks in the chain, and the ratio 
$\cc(\Qt;M)/\Sl{\Qt}^\alpha$ for a parameter $\alpha>0$.
\label{fig:CC}}
\vspace{1ex}
\end{figure}

However, if the algorithm does not have sufficient parallelism for the
PMH or is too irregular to load balance, we would expect it take
longer. Furthermore, since the number of processors assigned to a task
by a SB scheduler depends on its size, unlike in the case of
work-stealing style schedulers, a work-span analysis of programs is
not an accurate indicative of their running time.  A more
sophisticated cost model that takes account of locality,
parallelism-space imbalances, and lack of parallelism at different
levels of recursion is necessary.

In the NP model~\cite[Defn. 3]{BlellochFiGi11}, this was quantified by
the {\bf effective cache complexity metric (ECC, $\bcce$)}. We provide
a new definition of this metric for the ND model.  ECC attempts to
capture the cost of load balancing the program on hypothetical machine
with \defn{machine with parallelism at most} $\alpha$ --- a PMH which
has at most $f_i \le (M_i/M_{i-1})^\alpha$ level-$(i-1)$ caches
beneath each level-$i$ cache for all $1\le i\le h$.

The metric assigns to each subtree of a spawn tree an estimate of its
complexity, measured in cache miss cost equivalents, when mapped to a
PMH by a SB scheduler. The estimate is based on its position in the
spawn tree and its cache complexity in the PCC metric. The metric has
two free parameters: $\alpha$ which represents the parallelism
available of a hypothetical machine, and $M$ which represents the size
of one of the caches in the hierarchy with respect to which the spawn
tree is being analyzed.

\begin{definition}[Effective Cache Complexity (ECC)]
\label{def:qhat-fire}

Let $\Qt$ be a task in the ND model.  Unroll the spawn tree of $\Qt$,
applying the DAG rewriting rules until all the leaves of the tree are
$M$-maximal. Regard all dataflow arrows (solid or dashed) between the
leaves to be dependencies (see~\figref{CC}).

\noindent The ECC of a $M$-maximal task $\Qt's$ is
 $\bcc{\Qt';M} = \cc(\Qt';M)$.\\

\noindent The ECC of $\Qt$ is  $\bcc{\Qt;\MM}$, where
  $\left\lceil\frac{\bcc{\Qt;\MM}}{\Sl{\Qt}^\alpha}\right\rceil = $

\begin{eqnarray*}
 \max\left\{ 
\begin{array}{lll}
 \max_{\chi \in chains(\Qt,M)}
\left\{ 
\sum_{\Qt_i\in\chi}
\left\lceil\frac{\bcc{\Qt_i;\MM;\kappa}}{\Sl{\Qt_i}^\alpha}\right\rceil
\right\}
 & \mbox{(depth dominated)}\\
\left\lceil\frac{\sum_{\Qt_i\in maxiamal(\Qt,M)} \bcc{\Qt_i;\MM}}{\Sl{\Qt}^\alpha}\right\rceil
 & \mbox{(work dominated)}
\end{array}
\right.
\end{eqnarray*}
where $chains(\Qt,M)$ represents the set of chains of dependence edges
between $M$-maximal tasks, $maximal(\Qt,M)$, in the spawn tree of $\Qt$.
\end{definition}

The work dominated term has the same denominator as the left hand side
and thus captures the total amount of cache complexity in the spawn
tree (summation over leaves).  The depth dominated term captures the
critical path for the SB scheduler.  The term
$\lceil{\bcc{\Qt;\MM}}/{\Sl{\Qt}^\alpha}\rceil$ is the proxy for span
in our analysis and we call it the \defn{effective depth} of the task
$\Qt$. The depth dominated ensures that the effective depth defined by
ECC for a task is at least the sum of the effective depths of all
$M$-maximal tasks along any chain between $M$-maximal tasks induced
by DAG rewriting with respect to the fire rules.
The definition of ECC is such that:
\begin{enumerate}
\item
 In the range $\alpha\in [0,\alpha_{max})$, for some
algorithm-specific constant $\alpha_{max}$, $\bcc{M} \leq c_U \cc(M)$
for all $M>M_U$, for some positive universal constants $c_U, M_U$.
\item
 On a machine with parallelism $\beta \leq \alpha_{max}-\epsilon$ for
some arbitrarily small positive constant $\epsilon$, the running time
of the SB scheduler is within a constant factor of the perfectly load
balanced scenario in equation~\ref{eq:lb} (see~\thmref{rt}).
\item For NP programs, it coincides with the definition in \cite{BlellochFiGi11}.
\end{enumerate}

\para{Parallelizability of an Algorithm.} 
For the above reasons, we refer to the $\alpha_{max}$ of an algorithm
as its \defn{parallelizability} just as in \cite{Simhadri13}.  The
greater the parallelizability of the algorithm, the more efficient it
is to schedule on larger machines. When the parallelizability of the
algorithm asymptotically approaches the difference between the work
and the span exponents of the algorithm, we call it \defn{reasonably
regular}.  For an input of size $N = n\times n$, TRS, Cholesky and 2D
Floyd-Warshall have work exponent $1.5$ and span exponent $0.5$, and
the difference between them is $1$. In many divide-and-conquer
algorithms, such as in \cite{BlellochGiSi10}, where the NP model does
not induce too many artificial dependencies, the parallelizability
exceeds that of largest shared memory machines available today. In
such algorithms SB schedulers have been empirically shown to be
effective at managing locality without compromising load balance, and
as a consequence, capable of outperforming work-stealing
schedulers \cite{SimhadriBlFi14}. However, this is not the case for
algorithms in~\secref{num}, which lose some parallelism when expressed
in the NP model.

For example, in the NP model, the parallelizability (w.r.t. cache size
$M$) of the cache-oblivious matrix multiplication is $\alpha_{max,MM}
= 1-\log_{M}(1+c_{MM})$ for some small constant $c_{MM}$
(see~\clmref{CC-MM} in Appendix \ref{app-calc}), which is as high as
it can be. We expect the parallelizability of nested parallel TRS
algorithm to be less than $\alpha_{max,MM}$. In fact, for an $n\times
n$ upper triangular $T$ and a right hand side $B$ of size $N = n\times
n$, the parallelizability the nested parallel TRS algorithm
in~\eqref{TRS-NP} which is $1-\log_{\min\{N/M,M\}}(1+c_{TRS})$ (see
~\clmref{CC-TRS} in Appendix \ref{app-calc}). This is smaller than the
parallelizability of matrix multiplication when $N/M < M$. Since L3
caches are of the order of 10MB, the reduced parallelism adversely
affects load balance even in large instances that are of the order of
gigabytes (also empirically observed in~\cite{SimhadriBlFi15}).  When
expressed in the ND model, the parallelizability of TRS improves.
This can be seen from the geometric picture
in~\figref{trs-lattice-dag} where the depth dominated term
corresponding to the longest chain has effective depth
$c(N^{0.5}/M^{0.5})M^{1-\alpha} + c'$, which is less than the work
dominated term when $\alpha < 1-\log_{M}(1+c_{TRS})$.  This is the
parallelizability of TRS in the ND model.  This is also the case for
other linear algebra algorithms including Cholesky and LU
factorizations.

\vspace{3pt}
\para{Running time analysis.}
The main result of this section is~\thmref{rt} which shows
that SB schedulers can make use of the extra parallelizability of
programs expressed in the ND model.
\begin{theorem}
\vspace{-3pt}
Consider an $h$-level PMH with $p_h$ processors where a level-$i$ cache
has size $M_i$, fanout $f_i$ and cache miss cost $C_i$. Let $\Qt$ be a
task such that $S(\Qt;\BB)>f_h\MM_{h-1}/3$ (the scheduler allocates the
entire hierarchy to such a task) with parallelizability $\alpha_{max}$
in the ND model. Suppose that $\alpha_{max}$ exceeds the parallelism
of the machine by a constant. The running time of $\Qt$ is no more than:
\begin{multline*}
  \frac{\sum_{j=0}^{h-1} \bcc{\Qt; \MM_j/3}\cdot C_j}{p_h}\cdot v_h,
  \ \ \textrm{where overhead}\ v_h\ \textrm{is}\ \\ 
  v_h=2\prod_{j=1}^{h-1} \left(\frac{1}{k} +
    \frac{f_j}{(1-k)(M_j/M_{j-1})^{\alpha'}}\right),
\end{multline*}
for some constant $0<k<1$, where $\alpha' = \min\set{\alpha_{max},1}$.
\label{thm:rt}
\end{theorem}

The theorem says that when the machine parallelism is no greater than
the parallelizability of the algorithm in the ND model, the algorithm
runs within a constant factor ($v_h$) of the perfectly load balanced
scenario in~\eqref{lb}.  Relating this theorem to the definition of
machine parallelism, we infer that for highly regular algorithms
considered in this paper, the SB scheduler can effectively use up to
$O(N^{1-c'}/M_{h-1})$ level-$(h-1)$ subclusters for some arbitrarily
small constant $c'<0$.

We prove this theorem using the notion of effective work, the
separation lemma (lemma~\ref{lemma:lvl-decomp}) and a work-span
argument based on effective depth as in~\cite{BlellochFiGi11}.
The \defn{latency added effective work} is similar to the effective
cache complexity, but instead of counting just cache misses at one
cache level, we add the cost of cache misses at each instruction.
The cost $\rho(x)$ of an instruction $x$ accessing location $m$ is
$\rho(x)=W(x)+C'_{i}$, where $W(x)$ is the work, and $C'_i = C_0 + C_1
+\dots + C_{i-1}$ is the cost of a cache miss if the scheduler causes
the instruction $x$ to fetch $m$ from a level-$i$ cache in the
PMH. The instruction would need to incur a cache miss at level-$i$ if
it is the first instruction within the unique maximal level-$i$ task
that accesses a particular memory location.  Using this
per-instruction cost, we define effective work $\laew{.}$ of a task
using structural induction in a manner that is deliberately similar to
that of $\bcc{.}$.

\begin{definition}[Latency added cost]
With cost $\rho$ assigned to instructions, the \defn{latency added
effective work} of a task $\Qt$, or a strand $\Qs$ nested inside a
task $\Qt$ (from which it inherits space declaration) is:
\noindent\rulestrand
\[
\laew{\Qs} = \Sl{\Qt}^\alpha \sum_{x\in\Qs} \rho(x).
\]
\noindent\ruletask  For task $\Qt$ of size between $M_i$ and $M_{i+1}$,
 the l.a.e.w. is $\laew{\Qt}$, where
  $\left\lceil\frac{\laew{\Qt}}{\Sl{\Qt}^\alpha}\right\rceil = $
\begin{eqnarray*}
 \max\left\{ 
\begin{array}{lll}
 \max_{\chi \in chains(\Qt,M)}
\left\{ 
\sum_{\Qt_i\in\chi}
\left\lceil\frac{\laew{\Qt_i}}{\Sl{\Qt_i}^\alpha}\right\rceil
\right\}
 & \mbox{(depth dominated)}\\
\left\lceil\frac{\sum_{\Qt_i\in maximal(\Qt,M)} \laew{\Qt_i}}{\Sl{\Qt}^\alpha}\right\rceil
 & \mbox{(work dominated)}
\end{array}
\right.
\end{eqnarray*}
where $chains(\Qt,M)$ represents the set of chains of dependence edges
between $M$-maximal tasks, $maximal(\Qt,M)$, in the spawn tree of $\Qt$.
\end{definition}

Because of the large number of machine parameters involved
($\{\MM_i$,$C_i\}_i$ etc.), it is undesirable to compute the latency
added work directly for an algorithm.  Instead, we will show that the
latency added effective work can be upper bounded by a sum of per
(cache) level machine costs $\biw{(i)}{\cdot}$ that can, in turn be
bounded by machine parameters and ECC of the algorithm.  For
$i\in[h-1]$, $\biw{(i)}{\Qt}$ of a task $\Qt$ is computed exactly like
$\laew{\Qc}$ using a different base case: for each instruction $x$ in
$\Qc$, if the memory access at $x$ costs at least $C'_i$, assign a
cost of $\rho_{i}(x) = C_{i}$ to that node.  Else, assign a cost of
$\rho_i(x) = 0$.  Further, we set $\rho_0(x) = W(x)$, and define
$\biw{(0)}{\Qc}$ in terms of $\rho_o(\cdot)$. It also follows from
these definitions that $\rho(x) =
\sum_{i=0}^{h-1}\rho_i(x)$ for all instructions $x$.

\begin{lemma}
\textbf{Separation Lemma}: On an $h$-level PMH, and for a parameter
$\alpha>0$, for a task $\Qt$ with size at least $M_{h-1}$, we have:
\[
 \left\lceil \frac{\laew{\Qb}}{\Sl{\Qt}^\alpha} \right\rceil
\leq
 \left\lceil \frac{\sum_{i=0}^{h-1} \biw{(i)}{\Qt}.}{\Sl{\Qt}^\alpha} \right\rceil
\]
\label{lemma:lvl-decomp}
\end{lemma}
\begin{proof}
  The proof is based on induction on the structure of the task
in terms of decomposition into strands and maximal tasks.
  For the base case of the induction, consider the strand $\Qs$ at the
  lowest level in the spawn tree. If $S(\Qs)$ denotes the space of the
  strand or the task immediately enclosing $\Qs$ from which it
  inherits space declaration, then by
  definition \begin{align*} \laew{\Qs}
  &= \left(\sum_{x\in \Qs} \rho(x)\right)\cdot \Sl{\Qs}^\alpha \leq \left(\sum_{x\in \Qs} \sum_{i=0}^{h-1}\rho_i(x)\right)\cdot \Sl{\Qs}^\alpha\\
  &= \sum_{i=0}^{h-1}\left(\sum_{x\in \Qs} \rho_i(x)\cdot \Sl{\Qs}^\alpha\right)
  = \sum_{i=0}^{h-1} \biw{(i)}{\Qs}.  \end{align*}

For a task $\Qt$ corresponding to a spawn tree $T$, the \emph{latency
added effective depth}
$\left\lceil{\laew{\Qt;\MM}}/{\Sl{\Qt}^\alpha}\right\rceil$ is either
defined by the work or the depth dominated term which is one of the
chains in $chains(T)$, the set of chains of level-$(h-1)$ maximal
tasks in the spawn tree $T$.  Index the work dominated term as the
$0$-th term and index the chains in $chains(\Qt,M)$ in some order
starting from $1$.  Suppose that of these terms, the term that
determines $\left\lceil{\laew{\Qt;\MM}}/{\Sl{\Qt}^\alpha}\right\rceil$
is the $k$-th term. Denote this by $T_k$, and the $r$-th summand in
the term by $T_{k,r}$.  Similarly, consider the terms for evaluating
each of $\biw{(l)}{\Qt}$ -- which are numbered the same way as in
$\laew{\Qt}$ -- and suppose that the $k_l$-th term (denoted by
$T^{(l)}_{k_l}$) on the right hand side determines the value of
$\biw{(l)}{\Qb}$.  Then,

\begin{equation*}
 \left\lceil \frac{\laew{\Qb}}{\Sl{\Qt}^\alpha} \right\rceil
 = \sum_{r\in T_k} T_{k,r}  \leq \sum_{r\in T_k} \sum_{l=0}^{h-1} T^{(l)}_{k,r},
\end{equation*}
where the inequality is an application of the inductive hypothesis.
Further, by the definition of  $T^{(l)}_{k_l}$ and $T^{(l)}_{k}$, we have

\begin{equation*}
\sum_{r\in T_k} \sum_{l=0}^{h-1} T^{(l)}_{k,r}
    \leq  \sum_{r\in T_{k_l}} \sum_{l=0}^{h-1} T^{(l)}_{k_l,r}
    = \left\lceil \frac{\sum_{l=0}^{h-1} \biw{(l)}{\Qt}}{\Sl{\Qt}^\alpha} \right\rceil,
\end{equation*}
  which completes the proof.
\end{proof}
With the separation lemma in place for the ND model, the proof of
\thmref{rt} follows from the two lemmas which we adapt from
\cite{BlellochFiGi11}. The first is a bound on the per level latency
added effective work term in terms of the effective cache complexity.
The second is a bound on the runtime in terms of the latency added
effective work using a modified work-span analysis akin to
 Brent's theorem.

\begin{lemma}
  Consider an $h$-level PMH and a task (or a strand) $\Qt$. 
  If $\Qt$ is scheduled on this PMH using a space-bounded
  scheduler with dilation $\sigma = 1/3$, then
$\laew{\Qt} \leq \sum_{i=0}^{h-1} \bcc{\Qt; \MM_i/3, \BB}\cdot C_i$.
  \label{lemma:leveldecompapp}
\end{lemma}
\begin{lemma}
  Consider an $h$-level PMH and a task with parallelizability with
  $\alpha_{max}$ that exceeds the parallelism of the PMH by a small
  constant. Let $\alpha' = \min\set{\alpha_{max},1}$.
  Let $N_i$ be a task or strand which has been
  assigned a set ${\cal U}_t$ of $q \leq g_i(S(N_i))$
  level-$(i-1)$ subclusters by the scheduler. Letting $\sum_{V\in{\cal
  U}_t} (1-\mu(V)) = r$ (by definition, $r\leq \lvert{\cal U}_t\rvert
  = q$), the running time of $N_i$ is at most:
\begin{align*}
\frac{\laew{N_i}}{rp_{i-1}}\cdot v_i &,
\ \ \textrm{where overhead}\ v_i\ \textrm{is}\ \\ 
v_i&=2\prod_{j=1}^{i-1} \left(\frac{1}{k} +
  \frac{f_i}{(1-k)(M_i/M_{i-1})^{\alpha'}}\right).
\end{align*}
for some constant $0<k<1$.
  \label{lem:runtime}
\end{lemma}
The proofs of these two lemmas follow the same arguments as
in~\cite{BlellochFiGi11} with minor, but straightforward, modifications
that account for the new definition of the ECC in the ND model.

\punt{
\begin{lemma}
  Consider an $h$-level PMH and a task (or a strand) $\Qt$. 
  If $\Qt$ is scheduled on this PMH using a space-bounded
  scheduler with dilation $\sigma = 1/3$, then
$\laew{\Qt} \leq \sum_{i=0}^{h-1} \bcc{\Qt; \MM_i/3, \BB}\cdot C_i$.
  \label{lemma:leveldecompapp}
\end{lemma}
\begin{lemma}
  Consider an $h$-level PMH and a task with parallelizability with
  $\alpha_{max}$ that exceeds the parallelism of the PMH by a small
  constant. Let $\alpha' = \min\set{\alpha_{max},1}$.
  Let $N_i$ be a task or strand which has been
  assigned a set ${\cal U}_t$ of $q \leq g_i(S(N_i))$
  level-$(i-1)$ subclusters by the scheduler. Letting $\sum_{V\in{\cal
  U}_t} (1-\mu(V)) = r$ (by definition, $r\leq \lvert{\cal U}_t\rvert
  = q$), the running time of $N_i$ is at most:
\begin{align*}
\frac{\laew{N_i}}{rp_{i-1}}\cdot v_i &,
\ \ \textrm{where overhead}\ v_i\ \textrm{is}\ \\ 
v_i&=2\prod_{j=1}^{i-1} \left(\frac{1}{k} +
  \frac{f_i}{(1-k)(M_i/M_{i-1})^{\alpha'}}\right).
\end{align*}
for some constant $0<k<1$.
  \label{lem:runtime}
\end{lemma}
The proofs of these two lemmas follow the same arguments as
in~\cite{BlellochFiGi11} with minor, but straightforward, modifications
that account for the new definition of the ECC in the ND model.

 Note that we did not use the fact that
  some of the components were work or cache complexities. The proof
  only depended on the fact that $\rho(x) = \sum_{i=0}^{h-1}
  \rho_i(x)$ and the structure of the composition rules.
  $\rho$ could have been  replaced with any other kind of work and
  $\rho_i$ with its decomposition.

Consider the {\bf rule for a leaf node} $\ell$ nested immediately
within a task $\Qt$ with size $\Sl{\Qt}$: a SB scheduler might assign
it up to $\Sl{\ell}^\alpha$.  Since the leaf node has no parallelism
it uses only one processor for $\cc(\ell ; \MM; \kappa)$
``steps''. However, we make it ``pay'' the cost of keeping
$\Sl{\ell}^\alpha$ processors idle for $\cc(\ell; \MM; \kappa)$
``steps''.  Consider the {\bf rule for the parallel block} $\Qb
= \Qt_1\|\Qt_2\|\dots\|\Qt_k$, which is shorthand for a subtree rooted
at node $\Qb$ that is a disjoint union of a few $\parallel$ operators
and subtrees representing tasks $\Qt_i$.  We take the space of a
parallel block to be the size of the task in which it is immediately
nested, say $\Qt$. A lower bound on the complexity of $\Qb$ is the sum
of the complexities of the tasks $\Qt_i$, which gives rise to the
``work dominated'' term. Another possibility is that one of the tasks
in the parallel block takes much longer than the rest, and therefore
defines the complexity of $\Qb$.  To account for this, we use the term
${\bcc{\Qt_i;\MM}}/{\Sl{\Qt_i}^\alpha}$ as a measure of
the \defn{effective depth} of $\Qt_i$, i.e., how long it takes when
assigned $\Sl{\Qt_i}^\alpha$ processors. The ``depth dominated'' term
proposes that the effective depth of $\Qt$ is at least as much as the
effective depth of any of $\Qt_i$.  The {\bf rule for tasks} is
similarly explained as the composition of effective depths.  Thus this
metric is at once a generalization of work, span, and metrics of
locality, all of which can be recovered by setting parameters in
$\bcc{\Qt;M}$ appropriately (e.g., work is $\widehat{Q}_{0}(\Qt;0)$.

\begin{definition}\label{def:qhat}
\vspace{-8pt}
For cache parameter $M$ and parallelism parameter $\alpha$,
the \defn{effective cache complexity} of a leaf $\ell$, parallel block
$\Qb$, or task $\Qt$ starting at cache state $\kappa$ is defined as:

\noindent\comprule{leaf} Let $\Qt$ be the nearest containing task of
leaf $\ell$
\vspace{-4pt}
\[
\bcc{\ell; \MM; \kappa} = \cc(\ell ; \MM; \kappa) \times \Sl{\Qt}^\alpha
\]
\vspace{-4pt}
\noindent\ruleparblock For $\Qb = \Qt_1\|\Qt_2\|\dots\|\Qt_k$ in task $\Qt$, 
\begin{eqnarray*}
\frac{\bcc{\Qb;M;\kappa}}{\Sl{\Qt}^\alpha} = 
 \max\left\{ 
\begin{array}{lll}
 \max_i
\left\{ \frac{\bcc{\Qt_i;\MM;\kappa}}{\Sl{\Qt_i}^\alpha}
\right\}
 & \mbox{(depth dominated)}\\
\frac{\sum_i \bcc{\Qt_i;\MM;\kappa}}{\Sl{\Qt}^\alpha}
 & \mbox{(work dominated)}
\end{array}
\right.
\end{eqnarray*}
\vspace{-4pt}
\noindent\ruletask For $\Qt = c_1;c_2;\dots;c_k$, 
$
\frac{\bcc{\Qt;\MM;\kappa}}{\Sl{\Qt}^\alpha} 
= 
\sum_{i=1}^{k}\frac{\bcc{c_i;\MM;\kappa_i}}{\Sl{c_i}^\alpha},\quad
$
where $\kappa_i$ if $\kappa_i = \emptyset$ is $\Sl{\Qt}>M$ and if
$\Sl{\Qt}\leq M$, $\kappa_i = \kappa \cup_{j=1}^{i} \loc(c_j)$, where
$\loc(c_j)$ is the set of locations touched by $c_j$.
\vspace{-4pt}
\end{definition}
}

\section{Related Work and Comparison}\label{sec:relWork}

\para{Nested Parallelism, Complexity and Schedulers.}  In the analysis
of NP computations, theory usually considers two metrics: \emph{time
  complexity} and \emph{cache complexity}.  While some theoretical
analyses often consider these metrics separately, in reality, the
actual completion time of a program depends on both, since the cache
misses have a direct impact on the running time. Initial analyses of
schedulers for the NP model, such as the randomized work-stealing
scheduler~\cite{BlumofeLe99}, were based only on the time complexity
metric. While such analysis serves as a good indicator of scalability
and load-balancing abilities, better analyses and new schedulers that
minimize both communication costs and load balance in terms of time
and cache complexities on various parallel cache configurations have
been studied
\cite{AcarBlBl00,BlellochGiMa99,ChowdhurySiBl13,BlellochFiGi11}.

A major advantage of writing algorithms in the NP and ND models is
that it exposes locality and parallelism at different scales, making
it possible to design schedulers that can exploit parallelism and
locality in algorithms at different levels of the cache
hierarchy. Many divide-and-conquer parallel cache-oblivious algorithms
that can can achieve theoretically optimal bounds on cache complexity
(measured for the serial elision), work and span exist
\cite{BlellochGiSi10, ColeRa11}. For these NP algorithms, schedulers
can achieve optimal bounds on time and communication costs.

Another advantage of the NP (and ND) algorithms is that despite being
processor- and cache-oblivious, schedulers execute these algorithms
well with minimal tuning; the bounds are fairly robust across cache
sizes and processor counts.  Tuning of algorithms for time and/or
cache complexity has several disadvantages: first, the code structure
becomes more complicated; second, the parameter space to explore is
usually of exponential size; third, the tuned code is non-portable,
i.e., separate tuning is required for different hardware systems;
fourth, the tuned code may not be robust to variations and noise in
the running environment.  Recent work by Bender et
al. \cite{BenderEbFi14} showed that loop based codes are not
cache-adaptive, i.e., when the amount of cache available to an
algorithm can fluctuate, which is usually the case in a real-world
environment, the performance of tuned loop tiling based can suffer
significantly.  However, many runtimes and systems
(e.g. Halide~\cite{Halide13}) that map algorithms such as dense
numerical algebra, stencils and memoization algorithms to parallel
machines rely heavily on tuning as a means to extracting performance.

\para{Futures, Pipelines and other Synchronization Constructs.}  The
limitations of the NP model in expressing parallelism is known in the
parallel programming community. Several approaches, such as futures
\cite{BakerHe77, FriedmanWi78} and synchronization variables
\cite{BlellochGiMa97}, were proposed to express more general classes
of parallel programs.

Conceptually, the \textit{future} construct lets a piece of
computation run in parallel with the context containing it. The
pointer to future can then be passed to other threads and synchronize
at a later point. Several papers have studied the complexity of
executing programs with futures.  Greiner and Blelloch
\cite{GreinerBl99} discuss semantics, cost models and effective
evaluations strategies with bounds on the time complexity. Spoonhower
et al.  \cite{SpoonHowerBlGi09} calculate tight bounds on the locality
of work-stealing in programs with futures. The bounds show that moving
from a strict NP model to programs with futures can make WS schedulers
pay significant price in terms of locality. To alleviate this problem,
Herlihy and Liu \cite{HerlihyLi14} suggest that the cache locality of
future-parallel programs with work-stealing can be improved by
restricting the programs to using ``well-structured futures'': each
future is touched only once, either by the thread that created it, or
by a later descendant of the thread that created it. However, it is
difficult to express the algorithms in our paper as well-structured
futures without losing parallelism or locality.  One of the main
reasons for this is that the algorithm DAGs we consider have nodes
with multiple, even $O(n)$, outgoing dataflows which can not be easily
translated into ``touch-once'' futures. Even if we were to express
such DAGs with touch-once futures, the resultant DAG might be
unnecessarily serialized.  We seek to eliminate such artificial loss
of parallelism with the ND model.  Further, the analysis of schedulers
for programs with futures is limited to work-stealing, which is a less
than ideal candidate for multi-level cache hierarchies. To the best of
our knowledge, no provably good hierarchy-aware schedulers for
future-parallel programs exist.

\textit{Synchronization variables} are a more general form of
synchronization among threads in ``computation DAG'' and can be used
to implement futures.  Blelloch et al. \cite{BlellochGiMa97} present
the \textit{write-once synchronization variable}, which is a variable
(memory location) that can be written by one thread and read by any
number of other threads.  The paper also discusses an online
scheduling algorithm for a program with ``write-once synchronization
variables'' with efficient space and time bounds on the CRCW PRAM
model with the fetch-and-add primitive.

Though futures or synchronization variables provide a more relaxed
form of synchronization among threads in a computation DAG thus
exposing more parallelism, there are some key technical differences
between these approaches and the ND model.
First, the future construct fails to address the concept of ``partial
dependencies''. A thread computing a future is ``parallel'', not
``partially parallel'', to the thread touching the future. The runtime
always eagerly creates both threads before the future is computed,
thus possibly wasting asymptotically more space and incurring
asymptotically more cache misses.  In contrast, the ``$\fire$''
construct allows the runtime the flexibility of creating ``sink''
tasks as required when partial dependencies are met.
Second, there is no existing work on linguistic and runtime support
for the recursive construction and refinement of futures over spawn
trees.  While many dataflow programming models have been studied and
deployed in production over the last four
decades~\cite{JohnstonHaPaMi04}, the automatic recursive construction
of dataflow over the spawn tree, which is crucial in achieving
locality in a cache- and processor-oblivious fashion, is a new and
unique feature of our model.
Third, there are algorithms whose maximal parallelism can be easily
realized using the ``$\fire$'' construct but not with futures. In the
ND model, it is easy to describe algorithms in which a source can fire
multiple sink nodes, and a sink node can depend on multiple
sources. Such algorithms with nodes that involve multiple incoming and
outgoing dataflow arrows pose problems in future-parallelism
models. For instance, programming the LCS algorithm using futures
without introducing artificial dependencies is very cumbersome.  To
eliminate artificial dependencies, this class of problems requires
futures to be touched by descendants of the siblings of the node whose
descendant created the future. That is: the touching thread may be
created before the corresponding future thread is created.  To the
best of our knowledge, there is no easy scheme to pass the pointer to
a future up and down the spawn tree.

Another closely related extension of the nested parallel model is
``pipeline parallelism''. Pipeline parallelism can be constructed by
either futures (e.g. \cite{BlellochRe97} who used it to shorten span)
or synchronization variables, or by some elegantly defined linguistic
constructs \cite{LeeLeSc13}.  The key idea in pipeline parallelism is
to organize a parallel program as a linear sequence of stages.  Each
stage processes elements of a data stream, passing each processed data
element to the next stage, and then taking on a new element before the
subsequent stages have necessarily completed their processing.
Pipeline parallelism cannot express all the partial dependence
patterns described in this paper. To allow the expression of arbitrary
DAGs, interfaces for ``parallel task graphs'' and schedulers for them
have been studied \cite{JohnsonDaHa96, AgrawalLeSu10b}. While in
principle they can be used to construct computation DAGs that contain
arbitrary parallelism, the work flow is more or less similar to
dataflow computation without much emphasis on recursion, locality or
cache-obliviousness. The same limitation is true of pipeline
parallelism as well.

\para{Other algorithms, systems and schedulers.}
Parallel and cache-efficient algorithms for dynamic programming have
been extensively studied
(e.g.\cite{ChowdhuryRa10,GalilPa94,MalekiMuMy14,TangYoKa15}).  These
algorithms illustrate algorithms in which it is necessary to have
programming constructs that can express multiple (even $O(n)$)
dataflows at each node without serialization~\cite{GalilPa94}. The
necessity of wavefront scheduling and designs for it have been studied
in~\cite{MalekiMuMy14,TangYoKa15}.

Dynamic scheduling in dense numerical linear algebra on shared and
distributed memories, as well as GPUs, has been studied in the MAGMA
and PLASMA~\cite{PLASMA}, DPLASMA~\cite{DPLASMA}, and
PaRSEC~\cite{PaRSEC} systems. The programming interface used for these
systems is DaGUE~\cite{DaGUE}, which is supported by hierarchical
schedulers in runtimes~\cite{HierDaGUE,PaRSEC}. The DaGUE interface is
a slight relaxation of the NP model that allows recursive composition
of task DAGs representing dataflow within individual kernels. However,
the interface does not capture the notion of partial
dependencies. When DAGs of smaller kernels are composed to define
larger algorithms, the dependencies are either total or null.  The
ability to compose kernels with partial dependency patterns is key to
the ND model.

The FLAME project~\cite{FLAME} project, and subsequently the Elemental
project~\cite{Poulson13}, provides a systematic way of deriving
recursions and data dependencies in dense linear algebra from
high-level expressions ~\cite{Bientinesi05}, and using them to
generate data flow DAG scheduling ~\cite{Chan02Thesis}. The method
proposed in these works can be adapted to find the partial dependence
patterns derived by hand in this paper.

The Galois system developed at UT Austin \cite{PingaliNgKu11} 
proposes a data-centric formulation of algorithms called ``operator
formulation''. This formulation was initiated for handling irregular
parallel computation in which data dependencies can change at runtime,
and for irregular data structures such as graphs, trees and sets. 
In contrast, our approach was motivated by more regular parallel
computations such as divide-and-conquer algorithms. 

\section*{Acknowledgements}
We thank Prof. James Demmel, Dr. Shachar Itzhaky,
Prof. Charles Leiserson, Prof. Armando Solar-Lezama
and Prof. Katherine Yelick for valuable discussions
and their support in conducting this research.
Yuan Tang thanks Prof. Xiaoyang Wang, the dean of Software
School and School of Computer Science at Fudan University
for general help on research environment.
We thank the U.S. Department of Energy, Office of Science,
Office of Advanced Scientific Computing Research (DoE ASCR),
Applied Mathematics and Computer Science Program,
grants DE-SC0010200, DE-SC-0008700, and AC02-05CH11231,
for financial support, along with DARPA grant HR0011-12-2-0016,
ASPIRE Lab industrial sponsors and affiliates Intel, Google,
Huawei, LG, NVIDIA, Oracle, and Samsung, and MathWorks.

\bibliographystyle{abbrv}
\bibliography{allpapers}

\newpage
\punt{
\appput{app}{Appendix}
\input{gap}
\input{paren}
\input{ws-pseudocode}
\section{More Algorithms in the ND Model}
\label{app-algos}
\para{LU with Partial Pivoting.}
A straightforward parallelization of the $2$-way divide-and-conquer
algorithm by Toledo \cite{Toledo97}, combined with a replacement of
the \proc{TRS} algorithm by our new ND \proc{TRS}, yields an optimal
LU with partial pivoting algorithm for an $n\times m$ matrix with time
(span) bound $O(m \log n)$, and serial cache bound $O(nm^2/B\sqrt{M} +
n m + (n m \log m)/B)$ in the ideal cache model~\cite{FrigoLePr99}
with a cache of size $M$ and block size $B$.

}

\appendix

\section{Cache Complexity Calculations}
\label{app-calc}
\begin{claim}
The parallelizability of the recursive matrix multiplication algorithm 
in the NP model is  $\alpha_{max,MM} = 1-\log_{M}(1+c_{MM})$ for some
small constant $c_{MM}$.
\label{clm:CC-MM}
\end{claim}
\begin{proof}
For multiplying $N = n\times n$ matrices (which takes $3N$ space):
\begin{align*}
& \bccM{3N;M}  \\
& =  c(3N)^\alpha +  \max\{4^\alpha, 8\}\cdot
\bccM{3N/4;M}\\
& =  c\cdot N^\alpha \left( 3^\alpha \right) + 8 \cdot \bccM{3N/4;M},
\quad \textrm{for}\quad \alpha<1.5\\
& =  c\left( \frac{12^\alpha}{8\cdot 3^\alpha - 12^\alpha} \right)
\cdot
\left( \frac{(3N)^{1.5}}{M^{1.5-\alpha}} - N^\alpha\right)
+ \bccM{M;M} \left(\frac{3N}{M}\right)^{1.5} \\
& =  c\left( \frac{12^\alpha}{8\cdot 3^\alpha - 12^\alpha} \right)
\cdot
\left( \frac{(3N)^{1.5}}{M^{1.5-\alpha}} - N^\alpha\right)
+ \frac{(3N)^{1.5}}{M^{1.5-\alpha}},
\end{align*}
\noindent
assuming for simplicity that $3N$ is a power-of-$2$ multiple of $M$.
Since $\cc_{MM}(N;M) = O({N^{1.5}}/{M^{0.5}})$, we have $\bccM{n;M}
\leq c_U \cc_M(3n;M)$ when $\alpha \leq 1-\log_{M}(1+c_{MM})$ for some
small constant $c_{MM}$. Therefore, $\alpha_{max,MM} = 1-\log_{M}(1+c_{MM})$
is the parallelizability of the algorithm in the NP model.
\end{proof}

\begin{claim}
The parallelizability of the recursive TRS algorithm in~\eqref{TRS-NP}
in the NP model is  $\alpha_{max,TRS} = 1-\log_{\min\{N/M,M\}}(1+c_{TRS})$ for some
small constant $c_{TRS}$.
\label{clm:CC-TRS}
\end{claim}
\begin{proof}
We have for $T$ upper triangular and
the right hand side $B$ of size $N = n\times n$ that is overwritten by $X$:
\begin{align*}
& \bccT{3N/2;M}  \\
& =  c\left(\frac{3N}{2}\right)^\alpha
 +  2\cdot \max\{4^\alpha, 2\}\cdot \bccT{3N/8;M}
 +  2\cdot \bccM{3N/4;M} \\
& = c\left(\frac{3N}{2}\right)^\alpha + 2k_{MM} \frac{(3N/4)^{1.5}}{M^{0.5}}
 +  2\cdot 4^\alpha\cdot \bccT{3N/8;M},\\
  &\quad\quad \textrm{for}\quad \alpha>0.5,\quad k_{MM}\ \textrm{is a constant assicated with matrix multiply} \\
& = c\left(\frac{3N}{2}\right)^\alpha\left(2^{1+\log_4{3N/2M}}-1\right)\\
& \quad\quad  + 2k_{MM} \frac{(3N/4)^{1.5}}{M^{0.5}}\left( (2\cdot4^{\alpha-1.5})^{1+\log_4{3N/2M}}-1\right)\\
&\quad\quad +  2\cdot (4^\alpha)^{\log_4 {3N/2M}}\cdot \bccT{M;M}, \\ 
& = c\left(\frac{3N}{2}\right)^\alpha\left(2\left(\frac{3N}{2M}\right)^{0.5}-1\right)\\
&\quad\quad  + 2k_{MM} \frac{(3N/4)^{1.5}}{M^{0.5}}\left(4^{\alpha-1}\left(\frac{3N}{2M}\right)^{\alpha-1}-1 \right)\\
&\quad\quad +  2 \left(\frac{3N}{2M}\right)^\alpha M, 
\end{align*}
assuming, for simplicity, that $3N$ is a power-of-$2$ multiple of $M$.
Comparing $\bccT{3N/2;M}$ with $\cc_{TRS}(3N/2;M)$, which is
$O((3N/2)^{1.5}/M^{0.5})$ for $3N/2>M$,
gives a parallelizability of $1-\log_{\min\{N/M,M\}}(1+c_{TRS})$.
\end{proof}

\end{document}